\documentclass[11pt]{amsart}
\usepackage{a4wide}
\usepackage{amsmath, amsfonts, amssymb, amscd, graphicx, amsthm, cite}
\usepackage{mathrsfs,url}
\usepackage{hyperref}
\usepackage{color}

\usepackage[applemac]{inputenc}
\theoremstyle{plain}

    \newtheorem{thm}{Theorem}
    \newtheorem{theorem}[thm]{Theorem}

    \newtheorem{lemma}[thm]{Lemma}

    \newtheorem{proposition}[thm]{Proposition}

    \newtheorem{corollary}[thm]{Corollary}

    \newtheorem{conjecture}[thm]{Conjecture}
    
\theoremstyle{definition}

    \newtheorem{definition}[thm]{Definition}

     \newtheorem{observation}[thm]{Observation}

\newcommand{\N}{{\mathbb N}}
\newcommand{\Z}{{\mathbb Z}}

\DeclareMathOperator{\Sym}{Sym}

\DeclareMathOperator{\Pol}{Pol}

\DeclareMathOperator{\HSP}{\mathsf{HSP}}

\DeclareMathOperator{\Con}{Con}
\DeclareMathOperator{\arity}{ar}

\newcommand{\alg}[1]{\mathbf{#1}}

\newcommand{\proj}[1]{\mathrm{pr}_{#1}}

\newcommand{\algA}{\alg{A}}
\newcommand{\algB}{\alg{B}}

\newcommand{\algD}{\alg{D}}
\newcommand{\algK}{\alg{K}}
\newcommand{\algL}{\alg{L}}
\newcommand{\algM}{\alg{M}}

\newcommand{\algU}{\alg{U}}

\DeclareMathOperator{\CSAT}{\mathsf{CSAT}}
\DeclareMathOperator{\CEQV}{\mathsf{CEQV}}

\DeclareMathOperator{\PolSAT}{\mathsf{PolSAT}}
\DeclareMathOperator{\PolEQV}{\mathsf{PolEQV}}

\DeclareMathOperator{\PolEqv}{\mathsf{PolEQV}}

\author{Michael Kompatscher}
\date{\today}
\title{CSAT and CEQV for nilpotent Mal'tsev algebras of Fitting length $> 2$}

\thanks{This paper was supported by grant No 18-20123S of the Czech Science Foundation (GA\v{C}R), grant UNCE/SCI/022 of
the Charles University Research Centre, and INTER-EXCELLENCE project LTAUSA19070 of the Czech Ministry of Education M\v{S}MT. The author further received funding from the European Research Council (ERC) under the European Union's Horizon 2020 research and innovation programme (grant agreement No 714532). The paper reflects only the authors' views and not the views of the ERC or the European Commission. The European Union is not liable for any use that may be made of the information contained therein.}

\begin{document}

\begin{abstract}
The circuit satisfaction problem $\CSAT(\algA)$ of an algebra $\algA$ is the problem of deciding whether an equation over $\algA$ (encoded by two circuits) has a solution or not. While solving systems of equations over finite algebras is either in P or NP-complete, no such dichotomy result is known for $\CSAT(\algA)$. In fact, Idziak, Kawa\l{}ek and Krzaczkowski \cite{IKK-intermediate} constructed examples of nilpotent Mal'tsev algebras $\algA$, for which, under the assumption of the exponential time hypothesis and an open conjecture in circuit complexity,
 $\CSAT(\algA)$ can be solved in quasipolynomial, but not polynomial time. The same is true for the circuit equivalence problem $\CEQV(\algA)$.

In this paper we generalize their result to all nilpotent Mal'tsev algebras of Fitting length bigger than 2. This not only advances the project of classifying the complexity of $\CSAT$ (and $\CEQV$) for algebras from congruence modular varieties, but we also believe that the tools developed are of independent interest in the study of nilpotent algebras.
\end{abstract}

\maketitle

\section{Introduction}
Solving equations over a given algebraic structure is one of the most fundamental problems in mathematics. In its various forms, it was one of the driving forces in developing new methods and entire new fields in algebra (e.g. Gaussian elimination, Galois theory, and algebraic geometry, to only mention a few). A big interest nowadays often lies in studying the problem from a computational perspective. While solving equations over infinite algebras can be undecidable (as it was for instance famously shown for Diophantine equations \cite{matijasevic-diophantine}), it can always be done in non-deterministic polynomial time in finite algebras. So for finite algebras, one of the main tasks is to distinguish algebras, for which solving equations is tractable (in P), hard (NP-complete), or possibly of some intermediate complexity.

Solving \emph{systems} of polynomial equations over finite algebras is closely connected to the constraint satisfaction problems (CSP) of finite relational structures: in fact, it is straightforward to see that solving such systems over an algebra $\algA = (A,f_1,\ldots,f_n)$ can be modeled as the CSP, whose constraint relations are the function graphs of the operations $f_1,\ldots,f_n$ and all constants (e.g. observed in \cite{LZ-PolSys}). Therefore the CSP dichotomy theorem, which was independently shown by Bulatov \cite{bulatov-dichotomy} and Zhuk \cite{zhuk-dichotomy-short,zhuk-dichotomy}, implies that also solving systems of equations over $\algA$ is either in P or NP-complete.

But, interestingly, not much is known for the version of the problem, in which we restrict the input to a bounded number of equations over $\algA$. In fact, even the base case, in which the input consists of \emph{one} polynomial equation is far from being understood. In the literature this problem is usually referred to as the \emph{polynomial satisfiability problem} (or \emph{equation solvability problem}) $\PolSAT(\algA)$.

Complete classifications of $\PolSAT(\algA)$ exists for some classes algebras, such as rings \cite{Horvath-EQSAT-nilpotent}, or lattices \cite{schwarz-CSAT-lattices}. The problem is also quite well understood for finite groups (see \cite{GR-EQSAT-groups, foldvari-semipattern, FH-PolSATgroups, weiss-PolSAT}), despite the situation not being completely resolved for solvable groups of Fitting length $2$ (Problem 1 in \cite{IKKW-solvablegroups}). 

Classifying the complexity of $\PolSAT(\algA)$ for arbitrary algebras $\algA$ seems to be out of reach at the moment. One of the main obstacles is, that the complexity of $\PolSAT(\algA)$ is sensitive to changes of the language: In \cite{HS-EQSAT-A4} it was shown that $\PolSAT$ of the alternating group $(A_4,\cdot,1,^{-1})$ is in P, but, if we also allow expressions using commutator brackets $[x,y] = x^{-1}y^{-1}xy$ in the input, the problem becomes NP-complete.

To resolve this problem, Idziak and Krzaczkowski suggested in \cite{IK-CSAT} to look at the circuit satisfaction problem $\CSAT(\algA)$ instead, in which the input equation is not encoded by polynomials, but by circuits over $\algA$. Then, unlike for $\PolSAT$, the complexity is stable under changes of the signature of $\algA$ up to polynomial equivalence, which allows for an efficient use of methods from universal algebra.

The same paper also introduced the circuit equivalence problem $\CEQV(\algA)$, which asks if two circuits over $\algA$ are equivalent, and initiated a systematic study of $\CSAT$ and $\CEQV$ for algebras from \emph{congruence modular varieties}. This class of algebras, encompasses many algebras of interest in algebra and computer science (such as lattices, Boolean algebras, groups, quasigroups, rings, fields, vector spaces, modules), while still exhibiting some nice regularity properties.

In many cases, the computational hardness of $\CSAT(\algA)$ (and $\CEQV(\algA)$) is then intimately connected to existence of circuits that allow to express conjunctions. For example, in a finite field $(\mathbb F_p,+,0,-,\cdot)$ the polynomial $t_n(x_1,\ldots,x_n) = x_1^{p-1} \cdot x_2^{p-1} \cdots x_n^{p-1}$ can model the conjunctions of $n$ Boolean variables, as it is equal to $0$, if there is an $x_i = 0$, and equal to $1$ else. Equations of the form $t_n(x_1-y_1,\ldots,x_n-y_n) = 1$ can then straightforwardly be used to encode the $p$-colouring problem, thus $\CSAT((\mathbb F_p,+,0,-,\cdot))$ is NP-complete. However, if we remove the multiplication $\cdot$ from the signature (and therefore the ability to form conjunctions), we are left with the problem of solving a linear equation over $\mathbb F_p$, which is clearly in P. The polynomials $t_n$ are examples of so-called \emph{absorbing} operations.

The results in \cite{IK-CSAT} link the existence of non-trivial absorbing polynomials, commutator theoretical properties, and the complexity of $\CSAT(\algA)$ and $\CEQV(\algA)$ for algebras from congruence modular varieties: For so called \emph{supernilpotent} algebras $\algA$, which do not have absorbing polynomials of arbitrary big arities, the problems are in P (for $\CSAT$ this was independently shown in \cite{kompatscher-EQSAT-supernilpotent}, and for $\CEQV$ already in \cite{AM-commutator}). On the other hand, \emph{non-nilpotent} algebras allow to efficiently construct circuits, which express non-trivial absorbing polynomial operations. These circuits were used to prove that the circuit equivalence problem of non-nilpotent Mal'tsev algebras is coNP-complete, and the circuit satisfiability problem NP-complete \cite{IKK-Maltsev} (and, modulo taking quotients, versions of these hardness results also hold in the general congruence modular case \cite{IK-CSAT}).

This left an intriguing gap for nilpotent, but not supernilpotent Mal'tsev algebras (see Problem 2 in \cite{IK-CSAT}). Although such algebras admit non-trivial absorbing polynomial operations of all arities, there are some known examples $\algA$, for which $\CSAT(\algA)$ and $\CEQV(\algA)$ are solvable in polynomial time (see \cite{IKK-examples, KKK-CEQV-2nilpotent}). An explanation for this counter-intuitive phenomenon is, that these non-trivial absorbing polynomials cannot \emph{efficiently} be represented by circuits.

For nilpotent algebras in general, this is ties to an open conjecture by Barrington, Straubing and Th{\'e}rien in circuit complexity,  which states that CC-circuits of bounded width (or equivalently, programs over solvable groups) require exponential size to express conjunctions \cite{BST-NUDFA}. This conjecture was shown in \cite{kompatscher-CCcircuits} to be equivalent to absorbing polynomials in nilpotent algebras having the following lower bound on their size:

\begin{conjecture}[\cite{BST-NUDFA, kompatscher-CCcircuits}] \label{conjecture:BST}
For every finite nilpotent Mal'tsev algebra $\algA$, there is a constant $0 < q \leq 1$ such that any circuit $C_n(x_1,\ldots,x_n)$ over $\algA$ that is not constant and absorbing, has size at least $2^{\Omega(n^q)}$.
\end{conjecture}

Under Conjecture \ref{conjecture:BST} we get following conditional result, which indicated that the problems might be easy for nilpotent Mal'tsev algebras: 

\begin{theorem}[\cite{kompatscher-CCcircuits}] \label{theorem:CCcircuits}
Let $\algA$ be a finite nilpotent Mal'tsev algebra. Under the assumption of Conjecture \ref{conjecture:BST}, $\CSAT(\algA)$ and $\CEQV(\algA)$ can be solved in quasipolynomial time $2^{O((\log{n})^t)}$ (for $t = q^{-1}+1$).
\end{theorem}

In \cite{IKK-intermediate} Idziak, Kawa\l{}ek and Krzaczkowski contrasted this theorem with examples of nilpotent Mal'tsev algebras $\algD[p_1,\ldots,p_k]$, obtained as an extensions of Abelian groups, in which there are non constant $n$-ary absorbing circuits of size $2^{O(n^{\frac{1}{k-1}})}$ with some additional nice properties. Assuming the \emph{exponential time hypothesis (ETH)}, this allowed them to prove also proper, quasipolynomial lower bounds for the complexity of $\CSAT(\algD[p_1,\ldots,p_k])$ and $\CEQV(\algD[p_1,\ldots,p_k])$, which indicates, together with Theorem \ref{theorem:CCcircuits}, that both $\CSAT(\algD[p_1,\ldots,p_k])$ and $\CEQV(\algD[p_1,\ldots,p_k])$ have NP-intermediate complexities.

We are going to generalize their result to all nilpotent Mal'tsev algebra of \emph{Fitting length $k \geq 2$} (where the Fitting length, is a straightforward generalization of the same concept for groups, and measures how far the algebra is from being supernilpotent). So, the main result of this paper states as follows:

\begin{theorem} \label{theorem:main}
Let $\algA$ be a nilpotent Mal'tsev algebra of Fitting length $k \geq 2$. Then, under the exponential time hypothesis (ETH), the deterministic complexity of $\CSAT(\algA)$ and $\CEQV(\algA)$ has a lower bound of $2^{o((\log{n})^{k-1})}$.
\end{theorem}

Theorem \ref{theorem:main} together with Theorem \ref{theorem:CCcircuits} then implies the following:

\begin{corollary} \label{corollary:main}
Let $\algA$ be a nilpotent Mal'tsev algebra of Fitting length $k > 2$. Then, under the ETH and Conjecture \ref{conjecture:BST},  $\CSAT(\algA)$ and $\CEQV(\algA)$ can be solved in quasipolynomial, but not polynomial time. So then, in particular
\begin{itemize}
\item $\CSAT(\algA)$ is NP-intermediate,
\item $\CEQV(\algA)$ is coNP-intermediate.
\end{itemize}
\end{corollary}

By Corollary \ref{corollary:main}, nilpotent Mal'tsev algebras of Fitting length $k = 2$ remain essentially the only algebras from congruence modular varieties, for which we do not have a (conditional) classification of the complexity of $\CSAT$ and $\CEQV$ (see the discussion in Section \ref{sect:open}).

Our proof of Theorem \ref{theorem:main} is based on commutator theory,  and in particular the concept of \emph{wreath product representation} of nilpotent algebras, which might have some other future applications in understanding algorithmic and structural aspects of nilpotent algebras. Using such wreath product representations, we construct circuits over $\algA$, which allow us to reduce 3-SAT (respectively p-COLOR) to $\CSAT(\algA)$ and $\CEQV(\algA)$ in subexponential time, proving Theorem \ref{theorem:main}. 

An alternative construction of such circuits using methods from tame congruence theory was found by Idziak, Kawa\l{}ek and Krzaczkowsi in \cite{IKK-Maltsev}.  In fact, they even constructed such \emph{polynomials}, proving that also $\PolSAT(\algA)$ and $\PolEQV(\algA)$ have intermediate complexity in the setting of Corollary \ref{corollary:main}. We would however like to point out that a preprint of the underlying paper \cite{Fittingpreprint1} pre-dates \cite{IKK-Maltsev}. We further believe that it can be beneficial to have different proof of the same result. In the case of this paper, we believe that wreath product representations might have future applications in studying nilpotent Mal'tsev algebras outside the scope of $\CEQV$ and $\CSAT$ (such as the subpower membership problem, or the finite basis problem).

This paper is structured as follows: In Section \ref{sect:background} we give an overview over necessary notions and results from universal algebra and commutator theory. In particular we define Fitting series and the Fitting length for general algebras $\algA$. In Section \ref{sect:clonoids} and \ref{sect:absorbing} we extend some known results about linearly closed clonoids, respectively the higher commutator. In Section \ref{sect:wreathproduct} we use the representation of nilpotent Mal'tsev algebras as wreath products of affine algebras, to characterize properties of certain congruences. We further construct the absorbing polynomials, which are necessary to finish the proof of Theorem \ref{theorem:main} in Section \ref{sect:proof}. In Section \ref{sect:open} we discuss some open problems.

\section{Background} \label{sect:background}
In this section, we discuss some necessary definitions and results from universal algebra, especially commutator theory. For general background in universal algebra we refer to \cite{bergman-universal-algebra}. For background in commutator theory (in congruence modular varieties) we refer to the book \cite{FM-commutator-theory}, and the papers \cite{AM-commutator, moorhead-commutator} (for the higher commutator).

An algebra $\algA$ is a pair $(A,(f_i^{\algA})_{i \in I})$, where $A$ is a set and $(f_i^{\algA})_{i \in I}$ is a family of finitary operation $f_i^{\algA} \colon A^{k_i} \to A$. The set $A$ is called the \emph{universe} of $\algA$, and the functions $(f_i^{\algA})_{i \in I}$ are called the \emph{basic operations} of $\algA$. The \emph{signature} of $\algA$ is the family $((f_i,k_i))_{i \in I}$, i.e. it consists of the \emph{operation symbols} $f_i$, together with their arity $\arity(f_i)=k_i$. In this paper, \emph{finite algebras} are algebras that have both finite universe and finite signature.

The \emph{polynomial operations} of $\algA$ are the finitary operations on $A$ that can be defined by compositions of basic operations and elements of $A$. We denote the set of all polynomial operations of $\algA$ by $\Pol(\algA)$, and the set of all polynomials of arity $k$ by $\Pol^k(\algA)$. If $\Pol(\algA) = \Pol(\algB)$ we say that $\algA$ and $\algB$ are \emph{polynomially equivalent}. %, if $\Pol(\algA) \subseteq \Pol(\algB)$, we say $\algB$ is a \emph{polynomial extension} of $\algB$.

The most common way to encode the polynomial operations of $\algA$ is by \emph{polynomials}, i.e. terms that use operation symbols from the signature of $\algA$ and constant symbols for elements of $A$. For example $p(x_1,x_2,x_3,x_4) =  (x_1 \cdot x_4 + (-3)) \cdot (3 \cdot x_2)+1$, is a polynomial over the ring of integers $\algA = (\Z,+,0,-,\cdot)$ that defines a 4-ary polynomial operation $p^{\algA} \colon \Z^4 \to \Z$.

However, in an effort to compress the input, one can also encode polynomial operations by \emph{circuits} over $\algA$, i.e. $A$-valued circuits, with a unique output gate, whose gates are labeled by the basic operation of $\algA$ (see e.g. \cite{IK-CSAT}). Polynomials can be regarded as those circuits, whose underlying digraph is a tree.

Occasionally, when it is clear from the context, we are not going to formally distinguish between an operation symbol $f$ and the corresponding basic operation $f^{\algA}$. Similarly we are sometimes also not going to distinguish between a polynomial $p$ (circuit $C$) and the induced polynomial operation $p^{\algA}$ (or $C^{\algA}$), but this should never cause any confusion.

For a fixed finite algebra $\algA$ the circuit satisfaction problem $\CSAT(\algA)$ and the circuit equivalence problem $\CEQV(\algA)$ are defined as\\

\textsc{Circuit satisfaction} $\CSAT(\algA)$\\
\textsc{Input:} Two circuits $C_1,C_2$ over $\algA$ with input gates $\bar x = (x_1,\ldots,x_n)$\\
\textsc{Question:} Is there a tuple $\bar a \in A^n$ such that $C_1^{\algA}(\bar a) = C_2^{\algA}(\bar a)$?\\

\textsc{Circuit equivalence} $\CEQV(\algA)$\\
\textsc{Input:} Two circuits $C_1,C_2$ over $\algA$ with input gates $\bar x = (x_1,\ldots,x_n)$\\
\textsc{Question:} Does $\algA \models C_1(\bar x) \approx C_2(\bar x)$, i.e., does $C_1^{\algA}(\bar a) = C_2^{\algA}(\bar a)$ hold for all $\bar a \in A^n$?\\

If we only allow polynomials in the input, then we obtain the polynomial satisfaction problem $\PolSAT(\algA)$ and polynomial equivalence problem $\PolEQV(\algA)$. Thus $\PolSAT(\algA)$ trivially reduces to $\CSAT(\algA)$, and $\PolEqv(\algA)$ to $\CEQV(\algA)$. 

A congruence $\alpha$ of $\algA$ is an equivalence relation on $A$ that is preserved by the operations of $\algA$. The set of all congruences of $\algA$ forms the congruence lattice $\Con(\algA)$, whose maximal element is $\mathbf 1_A = A\times A$, and whose minimum is $\mathbf 0_A = \{(x,x) \colon x \in A\}$.  For two tuples $\bar a, \bar b \in A^n$ of the same length, we are going to write $(\bar a, \bar b) \in \alpha$, if the $(a_i, b_i) \in \alpha$, for all indices $i=1,\ldots,n$.

Algebras in this paper will always be from \emph{congruence modular varieties}, i.e. from varieties, in which every algebra has a modular congruence lattice. In fact, we are only going to study algebras which have a \emph{Mal'tsev term} $m$, i.e. a ternary term satisfying the identities $m(y,x,x) \approx m(x,x,y) \approx y$. For short, we call such algebras \emph{Mal'tsev algebras}. It is well-known that a variety has a Mal'tsev term if and only if it is \emph{congruence permutable}. In particular this implies that every Mal'tsev algebra generates a congruence modular variety.  The opposite implication is not true in general, but holds in the setting of nilpotent algebras \cite{FM-commutator-theory}. Thus, for simplicity we are going to refer to such algebras simply as \emph{nilpotent Mal'tsev algebras}, instead of the lengthy (albeit more general) ``nilpotent algebras from congruence modular varieties''.

\subsection{Commutator theory}
The commutator is an additional operation on the congruence lattice $\Con(\algA)$, that allows us to generalize concepts like Abelianess, nilpotence or solvability from group theory. In the following, higher dimensional form, the commutator was first introduced by Bulatov in \cite{bulatov-commutator}; but it was studied for the binary case $n=2$ long before that \cite{smith-maltsev, FM-commutator-theory}.

\begin{definition} \label{definition:commutator}
Let $\algA$ be an algebra and $\alpha_1,\ldots,\alpha_n, \delta \in \Con(\algA)$ congruences. We then say that \emph{$\alpha_1,\ldots,\alpha_{n-1}$ centralize $\alpha_n$ modulo $\delta$}, and write $C(\alpha_1,\ldots,\alpha_{n-1},\alpha_n;\delta)$, if for all arities $k_1,k_2,\ldots, k_n\geq 1$, for all polynomial operations $p(\bar x_1,\ldots,\bar x_n) \in \Pol^{k_1+\cdots+k_n}(\algA)$ and all tuples $\bar a_i, \bar  b_i \in A^{k_i}$ such that 
\begin{itemize}
\item $(\bar a_i, \bar b_i) \in \alpha_i$ for all $i = 1,\ldots,n$, and
\item $(p(\bar y_1, \ldots, \bar y_{n-1}, \bar a_n), p(\bar y_1, \ldots, \bar y_{n-1}, \bar b_n)) \in \delta$,\\ for all $(\bar y_1,\ldots, \bar y_{n-1}) \in \prod_{i = 1}^{n-1} \{ \bar a_i, \bar b_i \} \setminus \{( \bar b_1,\ldots, \bar b_{n-1}) \}$,
\end{itemize}
it also holds that
$(p(\bar b_1, \ldots, \bar b_{n-1}, \bar a_n), p(\bar b_1, \ldots, \bar b_{n-1}, \bar b_n)) \in \delta.$\\
The \emph{(higher) commutator} $[\alpha_1,\ldots,\alpha_n]$ is the smallest congruence $\delta$ such that $C(\alpha_1,\ldots,\alpha_n;\delta)$ holds.
\end{definition}
In groups, $[\alpha_1,\ldots,\alpha_n]$ is the congruence given by the commutator subgroup $[N_1,N_2,\ldots,N_n] = [N_1,[N_2,\ldots,[N_{n-1},N_n]\ldots]]$, where $N_i = [e]_{\alpha_i}$ is the normal subgroup corresponding to $\alpha_i$. Similarly, in rings, $[\alpha_1,\ldots,\alpha_n]$ is given by the product of the ideals corresponding to the congruences $\alpha_i$. 

Using the commutator, we define the following generalizations of solvability and nilpotence from group theory:

\begin{definition} \label{definition:nilpotent}
For $n \in \N$, we call a congruence $\alpha \in \Con(\algA)$
\begin{enumerate}
\item \emph{$n$-solvable} if $\alpha^{[n]} = \mathbf 0_A$ in the series defined by $\alpha^{[0]} = \alpha$ and $\alpha^{[i+1]} = [\alpha^{[i]}, \alpha^{[i]}]$,
\item \emph{$n$-nilpotent} if $\alpha^{(n)} = \mathbf 0_A$ in the series defined by $\alpha^{(0)} = \alpha$ and $\alpha^{(i+1)} = [\alpha, \alpha^{(i)}]$,
\item \emph{$n$-supernilpotent}, if $[\alpha,\ldots,\alpha]_{n} = [\underbrace{\alpha,\ldots,\alpha}_{n+1 \text{ times}}] = \mathbf 0_A$.
\end{enumerate}
Further $\alpha \in \Con(\algA)$ is called \emph{solvable (nilpotent, supernilpotent)}, if there is an $n$ such that $\alpha$ is \emph{$n$-solvable ($n$-nilpotent, $n$-supernilpotent)}.
The algebra $\algA$ is called ($n$-) \emph{solvable (nilpotent, supernilpotent)}, if the full congruence $\mathbf 1_A$ is ($n$-) \emph{solvable (nilpotent, supernilpotent)}.
\end{definition}

We want to stress that, although nilpotence and supernilpotence agree in groups and rings, the same is not true in general Mal'tsev algebras (as e.g. demonstrated by the loop in \cite{vaughanlee-nilloop}). The higher commutator has the following monotonicity properties (see e.g. \cite{AM-commutator}):
\begin{itemize}
\item[ (HC1)] $[\alpha_1,\ldots,\alpha_n] \leq \bigwedge_{i = 1}^n \alpha_i$
\item[ (HC2)] If $\alpha_i \leq \beta_i$ for all $i$, then $[\alpha_1,\ldots,\alpha_n] \leq [\beta_1,\ldots,\beta_n]$
\item[ (HC3)] $[\alpha_1,\ldots,\alpha_n] \leq [\alpha_2,\ldots,\alpha_n]$
\end{itemize}
In congruence modular variety, furthermore the following properties hold \cite{moorhead-commutator}: 
\begin{itemize}
\item[ (HC4)] $[\alpha_1,\ldots,\alpha_n] = [\alpha_{\pi(1)},\ldots,\alpha_{\pi(n)}]$ for all permutations $\pi \in \Sym(n)$
\item[ (HC5)] $[\alpha_1,\ldots,\alpha_n] \leq \beta$ if and only if $C(\alpha_1,\ldots,\alpha_n;\beta)$
\item[ (HC6)] $[\alpha_1/\beta,\ldots,\alpha_n/\beta] = ([\alpha_1,\ldots,\alpha_n]\lor \beta)/\beta$ in $\algA/\beta$
\item[ (HC7)] $[\alpha_1,\ldots\alpha_{i-1},\bigvee \Gamma,\alpha_{i+1},\ldots,\alpha_n] = \bigvee_{\gamma \in \Gamma} [\alpha_1,\ldots\alpha_{i-1},\gamma,\alpha_{i+1},\ldots,\alpha_n]$, for all $i=1,\ldots,n$
\item[ (HC8)] $[\alpha_1,\ldots\alpha_{i-1}, [\alpha_i,\alpha_{i+1},\ldots,\alpha_n]] \leq [\alpha_1,\ldots\alpha_{i-1}, \alpha_i,\alpha_{i+1},\ldots, \alpha_n]$ for all $i=2,\ldots,n-2$
\end{itemize}

Instead of working with Definition \ref{definition:commutator} itself, we are often going to use a characterization of the commutator in Mal'tsev algebras by Aichinger and Mudrinski. For this we need to define absorbing operations:

\begin{definition}An operation $c\colon A^n\to A$ is called \emph{absorbing at} $\bar a = (a_1,\ldots,a_n) \in A^n$, if it satisfies the identities $c(a_1,x_2,\ldots,x_n) \approx c(x_1,a_2,x_3,\ldots,x_n) \approx \cdots \approx c(x_1,x_2,\ldots,a_n) \approx c(\bar a)$. In the special case that $\bar a = (0,0,\ldots,0)$, and $c(\bar a) = 0$ for some fixed $0 \in A$, the operation $c$ is called \emph{$0$-absorbing}. 
\end{definition}

Then the higher commutator can be described as follows:

\begin{theorem}[Lemma 6.9 in \cite{AM-commutator}] \label{theorem:absorbing}
Let $\algA$ be a Mal'tsev algebra, and $\alpha_1,\ldots,\alpha_n \in \Con(\algA)$. Then $[\alpha_1,\ldots,\alpha_n]$ is generated as a congruence by the set
$$S = \{(c(\bar a),c(\bar b)) \colon \bar a, \bar b \in A^n, (a_i, b_i) \in \alpha_i \text{ for all } i, \text{ and } c \in \Pol^n(\algA) \text { is absorbing at } \bar a \}.$$
\end{theorem}

In particular, Theorem \ref{theorem:absorbing} was used in \cite{AM-commutator} to prove the following characterization of supernilpotent algebras:

\begin{theorem}[Proposition 7.7. in \cite{AM-commutator}] \label{theorem:absorbing2}
Let $\algA$ be a Mal'tsev algebra. Then $[\mathbf 1_A,\ldots,\mathbf 1_A]_k = \mathbf 0_A$ if and only if every $0$-absorbing polynomial of $\algA$ of arity $l>k$ is constant.
\end{theorem}

Furthermore, a finite Mal'tsev algebras is supernilpotent if and only if it is a direct product of nilpotent Mal'tsev algebras of prime power sizes \cite{kearnes-spectra}. This result was generalized to the following characterization of supernilpotent congruences by Mayr and Szendrei:

\begin{theorem}[Theorem 4.2. in \cite{MS-algebras-from-congruences}] \label{theorem:mayrdecomposition}
Let $\algA$ be a finite algebra such that $\HSP(\algA)$ omits type $\mathbf 1$. Then a congruence $\alpha \in \Con(\algA)$ is supernilpotent if and only there are $\beta_1,\ldots,\beta_k \leq \alpha$ satisfying the following conditions:
\begin{enumerate}
\item $\alpha$ is nilpotent
\item $\beta_1 \land \cdots \land \beta_k = 0$
\item $(\beta_1 \land \cdots \land \beta_{i-1}) \circ \beta_i = \beta_i \circ (\beta_1 \land \cdots \land \beta_{i-1})  = \alpha$ for every $i>1$
\item For every $i$ there is a prime $p_i$, such that the size of every block of $\alpha/\beta_i$ in $\algA/\beta_i$ is a power of $p_i$.
\end{enumerate}
\end{theorem}

Theorem \ref{theorem:mayrdecomposition} in particular applies to finite Mal'tsev algebras $\algA$.

\subsection{Fitting series}
In this section we generalize the group theoretic notion of Fitting series to arbitrary algebras, and make some basic observations. 

\begin{definition} \label{definition:fitting}
Let $\algA$ be an algebra, and $\alpha \in \Con(\algA)$. Then:
\begin{itemize}
\item A sequence of congruences $\mathbf 0_A = \alpha_0 \leq \alpha_1 \leq \cdots \leq \alpha_m = \alpha$ is called a \emph{Fitting series} of $\alpha$, if for every $i = 1,\ldots, m$ there is an integer $j_i$, such that $[\alpha_i,\ldots,\alpha_i]_{j_i} \leq \alpha_{i-1}$ (so $\alpha_i$ is supernilpotent modulo $\alpha_{i-1}$).
\item The \emph{Fitting length of $\alpha$} is the smallest $m \in \N$, such that $\alpha$ has a Fitting series of length $m$. If $\alpha$ has no Fitting series, we say it has Fitting length $\infty$. 
\item Fitting series (Fitting length) of $\mathbf 1_A$ are also called \emph{Fitting series (Fitting length)} of $\algA$.
\end{itemize}
Algebras $\algA$ of Fitting length $\leq m$ were also called \emph{$m$-step supernilpotent} in \cite{IKK-intermediate}. 
\end{definition}

Note that, since in groups nilpotence and supernilpotence coincide, a Fitting series of a group is also a Fitting series in the sense of Definition \ref{definition:fitting}. (We remark that, for this reason, one could equally make the point of defining Fitting series to be series of congruences in which every element is \emph{nilpotent} modulo its predecessor. We are however not aware of any previous such use of the notion in the literature.  Considering structural results such as Theorem \ref{theorem:mayrdecomposition}, we think that Definition \ref{definition:fitting} is the more useful and hence more fitting generalization of ``Fitting series'' from group theory.)

As in group theory, we can describe algebras of finite Fitting length by their \emph{upper} and \emph{lower} Fitting series.

\begin{lemma} \label{lemma:lowerF}
Let $\algA$ be an algebra such that $\Con(\algA)$ is finite. Then $\algA$ has Fitting length $m < \infty$ if and only if the series defined by $\gamma_0 = \mathbf 1_A$, and $\gamma_{i+1} = \bigcap_{n \in \N} [\gamma_i,\ldots,\gamma_i]_{n}$ reaches $\mathbf 0_A$ after exactly $m$ steps. We call  $\gamma_0 \geq \gamma_1 \geq \gamma_2 \geq \cdots$ the \emph{lower Fitting series} of $\algA$.
\end{lemma}

\begin{proof}
Since $\algA$ has only finitely many congruences (and by (HC3)), $\gamma_{i+1} = [\gamma_i,\ldots,\gamma_i]_l$ for a big enough $l$. Thus, if the lower Fitting series reaches $\mathbf 0_A$ after $m$ many steps, then $\mathbf 0_A = \gamma_{m} \leq \gamma_{m-1} \leq \cdots \leq \gamma_0 = \mathbf 1_A$ is a Fitting series of length $m$.

To show that $m$ is minimal, let $\mathbf 0_A = \alpha_0 \leq \alpha_1 \leq \cdots \leq \alpha_k = \mathbf 1_A$  be an arbitrary Fitting series. We then prove by induction on $i=0,\ldots,k$ that $\gamma_{i} \leq \alpha_{k-i}$. For $i = 0$ this is clearly true. For an induction step $i \to i+1$ note that, by (HC2): $[\gamma_i,\ldots,\gamma_i]_l \leq [\alpha_{k-i},\ldots,\alpha_{k-i}]_l$ for every $l$. For big enough $l$ we get $\gamma_{i+1} = [\gamma_i,\ldots,\gamma_i]_l \leq [\alpha_{k-i},\ldots,\alpha_{k-i}]_l \leq \alpha_{k-i-1}$. We conclude that $\gamma_{k} \leq \alpha_{0} = \mathbf 0_A$ and therefore $m \leq k$.
\end{proof}

\begin{observation} \label{observation:fittingcongruence}
Assume that $\algA$ generates a congruence modular variety. Then it follows directly from (HC7) and (HC3) that $[\alpha \lor \beta,\ldots,\alpha \lor \beta]_{k+l} \leq [\alpha,\ldots,\alpha]_k \lor [\beta,\ldots,\beta]_l$ for all $\alpha,\beta \in \Con(\algA)$. As a consequence, the join of two supernilpotent congruences is again supernilpotent. So if $\Con(\algA)$ is finite, this implies that there is a maximal supernilpotent congruence.
\end{observation}

\begin{definition}
Let $\algA$ be an algebra from a congruence modular variety, such that $\Con(\algA)$ is finite. Then, the \emph{Fitting congruence} $\lambda$ of $\algA$ is the maximal supernilpotent congruence.
\end{definition}

Note that the Fitting congruence generalizes the concept of the \emph{Fitting subgroup}, which is the maximal (super-)nilpotent normal subgroup of a group. The existence of such a Fitting congruence allows us to generalize upper Fitting series from group theory:

\begin{lemma} \label{lemma:upperF}
Assume $\algA$ is from a congruence modular variety and $\Con(\algA)$ is finite. We then set $\lambda_0 = \mathbf 0_A$ and define $\lambda_{i+1} \geq \lambda_{i}$ to be the maximal congruence that is supernilpotent modulo $\lambda_i$ for every $i \in \N$. We call $\mathbf 0_A = \lambda_0 \leq \lambda_1 \leq \cdots$ the \emph{upper Fitting series}. Then $\algA$ has Fitting length $m < \infty$ if and only if the lower Fitting series reaches $\mathbf 1_A$ after exactly $m$ steps.
\end{lemma}

\begin{proof}
By Observation \ref{observation:fittingcongruence}, the upper Fitting series is well-defined. By its definition, for every $i$ there is a $j_i$ such that $[\lambda_i,\ldots,\lambda_i]_{j_i} \leq \lambda_{i-1}$. Thus, if $\lambda_m = \mathbf 1_A$, then the Fitting length of $\algA$ is bounded by $m$. 

For the other direction, assume that there is some other Fitting series $\mathbf 0_A = \alpha_0 \leq \alpha_1 \leq \cdots \leq \alpha_k = \mathbf 1_A$. We claim that $\alpha_i \leq \lambda_i$ for every $i= 0,\ldots,k$. For $i = 0$ this is clearly true. For an induction step $i \to i+1$, note that $[\alpha_{i+1},\ldots,\alpha_{i+1}]_l \leq \alpha_i \leq \lambda_i$ for some $l$, so $\alpha_{i+1}$ is supernilpotent modulo $\lambda_{i}$. But $\lambda_{i+1}$ is the join of all congruences that are supernilpotent modulo $\lambda_i$, thus $\alpha_{i+1} \leq \lambda_{i+1}$.
\end{proof}

We can further generalize the fact that a group is solvable if and only if it has finite Fitting length:

\begin{lemma} \label{lemma:solvableF}
Every solvable algebra $\algA$ has finite Fitting length (which is bounded by the degree of solvability).
If $\algA$ generates a congruence modular variety, also the converse holds.
\end{lemma}

\begin{proof}
For solvable $\algA$ the derived series $\mathbf 1_A^{[0]}, \mathbf 1_A^{[1]},\mathbf 1_A^{[2]},\ldots, \mathbf 1_A^{[m]} = \mathbf 0_A$ is clearly a Fitting series; so $\algA$ has finite Fitting length.

On the other hand, assume that $\algA$ generates a modular variety and has a finite Fitting series $\mathbf 0_A = \alpha_0 \leq \alpha_1 \leq \cdots \leq \alpha_m = \mathbf 1_A$. By (HC3) there is an integer $l$ such that $[\alpha_i,\ldots,\alpha_i]_{2^l} \leq \alpha_{i-1}$ for every $i$. From property (HC8) of the higher commutator it follows that $\alpha_i^{[l]} \leq [\alpha_i,\ldots,\alpha_i]_{2^l} \leq \alpha_{i-1}$. Therefore $\algA$ is solvable.
\end{proof}

Note that Lemma \ref{lemma:lowerF}, \ref{lemma:upperF} and \ref{lemma:solvableF} can be generalized to characterize \emph{congruences} of finite Fitting length.
%We repeatedly used that for algebras $\algA$ with finite congruence lattice, there exists an integer l = $l(\algA)$ such that the series of higher commutators $[\theta]_0, [\theta]_1,[\theta]_2,\ldots$ stabilizes after $l$ many steps. For finite algebras this leads to the question: 
%
%\begin{question}
%Let $\algA$ be a finite algebra. Is there a 'nice' bound on $l(\algA)$ (which only depends on $|A|$ and the maximal arity of $\algA$)?
%\end{question}

%If $\algA$ has a Mal'tsev term, I think $l(\algA)$ should be similar to Erhard Aichinger's bound on the degree of supernilpotence in \cite{aichinger-spectrum}. More generally, we could try to prove a bound for finite algebras omitting type 1, using the congruence algebra construction from \cite{ms-congruencealgebras}.
%
%\begin{question}
%I'm slightly unhappy with the notation we use at the moment. '$k$-step supernilpotent' can be easily confused with $\algA$ satisfying $[\mathbf 1_A,\ldots,\mathbf 1_A]_k = 0$. Also 'Fitting series' has some ambiguity to it: either the quotients $\alpha_{i+1}/\alpha_i$ are supernilpotent (like in this writeup) or nilpotent (as for example used in the definition of 'Fitting congruence' in Chapter 10 of \cite{freese-mckenzie}). \\
%Are there any better ideas?
%\end{question}

\subsection{Nilpotent Mal'tsev algebras}

By Definition \ref{definition:nilpotent} an algebra $\algA$ is ($n$-)nilpotent, if the series $\mathbf 1_A \geq [\mathbf 1_A,\mathbf 1_A] \geq [\mathbf 1_A,[\mathbf 1_A,\mathbf 1_A]] \geq \cdots$ reaches $\mathbf 0_A$ after $n$-many steps. It is not hard to see that this is equivalent to the existence of any \emph{central series} of length $n$, i.e. any series of congruences $\mathbf 0_A  = \alpha_0 \leq \alpha_1 \leq \cdots \leq \alpha_n = \mathbf 1_A$, such that $[\mathbf 1_A,\alpha_{i+1}] \leq \alpha_i$ for all $i = 1,\ldots,n$. Finite nilpotent algebras show some very regular behaviour in congruence modular varieties. In particular, they always have a Mal'tsev term \cite[Chapters 6 \& 7]{FM-commutator-theory}, therefore we are only going to refer them as nilpotent Mal'tsev algebras from now on.

A Mal'tsev algebra $\algA$ is 1-nilpotent (or \emph{Abelian}) if and only if it is \emph{affine}, i.e. polynomially equivalent to a module \cite{smith-maltsev,herrmann-affine}. Here, a module is an algebra $(A,+,0,-, (r)_{r \in R})$, where $(A,+,0,-)$ is the underlying Abelian group and every ring element $r \in R$ is considered as a unary operation $r(x) = r \cdot x$. Therefore all polynomial operations of an affine algebra $\algA$ are affine operations $p(x_1,\ldots,x_n) = \sum_{i = 1}^n r_i x_i + c$ for some scalars $r_i \in R$ and a constant $c \in A$. In particular, any Mal'tsev polynomial of $\algA$ must be of the form $m(x,y,z) = x - y + z$.

General nilpotent Mal'tsev algebras can be written as the \emph{wreath products} of affine algebras, which we will define below:

\begin{definition}
Let $\algL = (L,(f^{\algL})_{f \in F})$ and $\algU = (U,(f^{\algU})_{f \in F})$ be two algebras of the same signature $F$, such that $\algL$ is affine. Further let $T = (\hat f)_{f \in F}$ be a family of operations $\hat f \colon U^{\arity(f)} \to L$. Then we define $\algA = \algL \otimes^{T} \algU$ as the algebra such that:
\begin{itemize}
\item its universe is $A = L \times U$
\item every basic operation on $\algA$ is given by:
\begin{align*}
\proj{U}(f^{\algA}(\bar x)) &= f^{\algU}(\proj{U}(\bar x)) \\
\proj{L}(f^{\algA}(\bar x)) &= f^{\algL}(\proj{L}(\bar x)) + \hat f(\proj{U}(\bar x)),
\end{align*}
where $\proj{U}(\bar x)$ and $\proj{L}(\bar x)$ denote the (coordinate-wise) projection of a tuple $\bar x \in A^n$ to $U^n$, or $L^n$ respectively, and $+$ is the addition on $\algL$.
\end{itemize}
\end{definition}

We are going to call $\algL \otimes^{T} \algU$ a \emph{wreath product}, since it is a special case of a wreath product according to the definition of VanderWerf in \cite{vanderwerf-thesis, vanderwerf-wreathproduct}. In Proposition 7.1. of \cite{FM-commutator-theory}, Freese and McKenzie showed that Mal'tsev algebras $\algA$ with center $\zeta$ can be written as the wreath product of $\algU = \algA/\zeta$, and an affine algebra $\algL$. Following their proof, it is not hard to see that the result holds for arbitrary central congruences $\zeta$:

\begin{theorem}[Proposition 7.1. in \cite{FM-commutator-theory}] \label{theorem:freese-mckenzie}
Let $\algA$ be a Mal'tsev algebra, $\zeta \in \Con(\algA)$ such that $[\mathbf 1_A,\zeta] = \mathbf 0_A$, and $\algU = \algA/\zeta$.  Then there is an affine algebra $\algL$ in the same signature and a family of operations $T$, such that $\algA \cong \algL \otimes^{T} \algU$. 
\end{theorem}

Note that $\zeta$ is the kernel of the projection $\proj{U}$. Theorem \ref{theorem:freese-mckenzie} directly implies the following:

\begin{corollary} \label{corollary:freese-mckenzie}
Let $\algA$ be a Mal'tsev algebra, let $\mathbf 0_A = \alpha_0 < \alpha_1 < \cdots < \alpha_n$ be a series of congruences such that $[\mathbf 1_A,\alpha_{i}] \leq \alpha_{i-1}$ for all $i=1,\ldots,n$, and let $\algU = \algA / \alpha_n$. Then there are affine algebras $\algL_1, \ldots, \algL_{n-1}$ in the same signature as $\algA$, and families of operations $T_1,\ldots, T_{n-1}$, such that 
\begin{align} \label{eq:wreath} \algA &\cong \algL_{1} \otimes^{T_1} (\cdots (\algL_{n-1} \otimes^{T_{n-1}} \algU)).
\end{align}
\end{corollary}

\begin{definition} \label{definition:wreathproduct}
For simplicity we abbreviate the representation of the Mal'tsev algebra $\algA$ in (\ref{eq:wreath}) by $\bigotimes_{i = 1}^n \algL_i \otimes \algU$, and call it a \emph{wreath product representation of $\algA$} with respect to $\mathbf 0_A = \alpha_0 < \alpha_1 < \cdots < \alpha_n$. Note that (if we identify the domains of $\algA$ and $\bigotimes_{i = 1}^n \algL_i \otimes \algU$) then, for every basic operation $f$, its projection to the component $L_j$ is of the form
\begin{align}\label{eq:piLi} \proj{L_i} (f^{\algA}(\bar x)) &= f^{\algL_i}(\proj{L_i}(\bar x)) + \hat f_i(\proj{(L_{i+1}\times \cdots \times L_n \times U)}(\bar x)),
\end{align}
for a function $\hat f_i \colon (L_{i+1} \times \cdots \times L_{n} \times U)^k \to L_i$.
\end{definition}

We remark that there is some freedom in how to pick a wreath product representation of $\algA$ for a series $\mathbf 0_A = \alpha_0 < \alpha_1 < \cdots < \alpha_n$, that stems from the choice of a constant $0^{L_i}$ in every affine component $\algL_i$ and, with it, the addition $x+y = m^{\algL_i}(x,0^{L_i},y)$. Thus, implicitly, every wreath product representation comes with a constant $0 = (0^{L_1},0^{L_2},\ldots,0^{L_n}) \in \prod_{i=1}^n L_i$. For a fixed such constant we can ensure that $\hat f_i$ and $f_i^{\algL_i}$ are uniquely determined by requiring that $f_i^{\algL_i}(0^{L_i},0^{L_i},\ldots,0^{L_i}) = 0^{L_i}$ for all $i$. For notational simplicity, we are not going to distinguish the zeros of the different affine compontents, i.e. $0 = (0,0,\ldots,0) \in \prod_{i=1}^n L_i$; this should however never cause any confusion.

If $\alpha_n = \mathbf 1_A$ we simply write $\algA = \bigotimes_{i = 1}^n \algL_i$. It is not hard to see that every Mal'tsev algebra of the form $\bigotimes_{i = 1}^n \algL_i$ is $n$-nilpotent. So by Corollary \ref{corollary:freese-mckenzie}, a Mal'tsev algebra is $n$-nilpotent if and only if, it is isomorphic to some wreath product $\bigotimes_{i = 1}^n \algL_i$.

Note that also for every polynomial operation $p \in \Pol(\algA)$, and every component $L_i$, there is a decomposition of $\proj{L_i} (p^{\algA})$ into $p^{\algL_i}$ and $\hat p_i \colon (L_{i+1} \times \cdots \times L_{n}\times U)^k \to L_i$ as in (\ref{eq:piLi}).  A well-known consequence of this, is that nilpotent Mal'tsev algebras define loop operations in the following sense: 

\begin{theorem}[Corollary 7.4. in \cite{FM-commutator-theory}] \label{theorem:loop}
Let $\algA$ be a nilpotent algebra with Mal'tsev polynomial $m \in \Pol(\algA)$, and let $0 \in A$. Then $x + y = m(x,0,y)$ is a loop multiplication with neutral element $0$. Furthermore also the left and right inverse $x / y$ and $x \backslash y$ are polynomials of $\algA$.
\end{theorem}

We would like to emphasize that the above loop operation '$+$' is in general neither commutative, nor associative, but we choose the additive notation, since this works well with wreath product representations. If e.g. $\bigotimes_{i = 1}^n \algL_i$ is a representation of $\algA$ (with respect to the constant $0$), then $\proj{L_i} (x +^{\algA} y)$ is equal to the sum $\proj{L_i}  (x) +^{\algL_i}  \proj{L_i} (y)$ plus some ``distortion'' function $\hat\phi(\proj{(L_{i+1}\times \cdots \times L_n)}(x,y))$.

Also absorbing polynomials behave nicely with respect to wreath product decomposition, by the following observation:

\begin{observation} \label{observation:absorbing}
Let $\algA = \bigotimes_{i=1}^n \algL_i \otimes \algU$ be a wreath product representation of $\algA$, and $p \in \Pol^k(\algA)$ be a $0$-absorbing polynomial of arity $k \geq 2$. Then, for every $i = 1,\ldots,n$:
$$\proj{L_i}  (p^{\algA}(\bar x)) = \hat p_i (\proj{(L_{i+1}\times \cdots \times L_n \times U)}(\bar x)),
$$
for an absorbing function $\hat p_i \colon (L_{i+1} \times \cdots \times L_{n}\times U)^k \to L_i$. This follows directly from the fact that $p^{\algL_i}$ must be a $k$-ary affine absorbing operation, and therefore constant.
\end{observation}

We define the \emph{exponent $\exp(\algL_i)$} of an affine algebra $\algL_i$ to be the exponent of the group reduct $(L_i,+,0,-)$ of $\algL_i$ (i.e. the smallest integer $n$ such that $n \cdot x = x + \cdots + x = 0$ for all $x \in L_i$). Most of the time, we are only going to consider wreath products, such that every affine $\algL_i$ has a prime exponent, so $(L_i,+,0,-) \cong \Z_p^m$, where $p$ is a prime. Let us call such wreath products \emph{elementary} (as the groups $(L_i,+,0,-)$ are elementary Abelian). It is enough to focus on elementary wreath products by the following observation:

%should we also mention definability?
\begin{observation} \label{observation:subdivision}
If $\algL$ is a module, and $\algM$ a submodule of $\algL$, then $\algL \cong \alg{M} \otimes (\algL / \alg{M})$. In particular, note that if $\algM = \{ p \cdot x \colon x \in L \}$ for a prime divisor $p$ of $|L|$, then $\exp(\algL / \alg{M}) = p$, and the unary polynomial $p \cdot x$ maps $L$ to $M \times \{0\}$. It is not hard to see that this is also true for affine algebra (wreath product decompositions are preserved under polynomial equivalence).

As a consequence, every wreath product representation $\bigotimes_{i=1}^n \algL_i \otimes \algU$ of a Mal'tsev algebra can be refined into an elementary wreath product, by decomposing every factor $\algL_i$ into elementary components as described above. In particular, every central series $\mathbf 0_A = \alpha_0 \prec \alpha_1 \prec \ldots \prec \alpha_n = \mathbf 1_A$, in which $\alpha_i$ covers $\alpha_{i-1}$ induces an elementary wreath product.
\end{observation}

At last, note that a wreath product also encodes other congruences (and corresponding quotient algebras) than just the series $\mathbf 0_A = \alpha_0 < \alpha_1 < \cdots < \alpha_n$ in the following sense:

\begin{observation} \label{observation:kernel}
Let $\algA = \bigotimes_{i=1}^n \algL_i \otimes \algU$ be a wreath product representation of $\algA$, and let $I \subseteq \{1,\ldots,n\}$, such that for every polynomial $p \in \Pol(\algA)$ and every $j \notin I$, $\proj{L_i} (p(\bar x))$ does not depend on $\proj{L_i}(\bar x)$ for all $i \in I$. Then, projecting $\algA$ to $\bigotimes_{j \notin I} L_j \otimes U$ is a well-defined homomorphism.
\end{observation}

\section{A result about linearly closed clonoids} \label{sect:clonoids}
In this section we discuss a result about polynomial clonoids, which we define as follows:
%, which is a slight generalization of the results in e.g. \cite[Lemma 5.2]{AM-zpq}, or \cite[Lemma 3.1.]{IKK-examples}).

\begin{definition}
Let $\algL$, $\algU$ be two algebras with universes $L$ and $U$ respectively. Let us call a set of operations $\mathcal C \subseteq \bigcup_{n \in \N} L^{U^n}$ a \emph{(polynomial) $(\algL,\algU)$-clonoid}, if 
\begin{itemize}
\item $f \circ (g_1,\ldots,g_n) \in \mathcal C$ for all $g_1,\ldots,g_n \in \mathcal C$ and $f \in \Pol(\algL)$
\item $g \circ (h_1,\ldots,h_m) \in \mathcal C$ for all $g \in \mathcal C$ and $h_1,\ldots,h_m \in \Pol(\algU)$
\end{itemize}
\end{definition}

So, in words, $\mathcal C$ is a $(\algL,\algU)$-clonoid, if it is closed under composition with polynomials of $\algL$ from the outside, and polynomials of $\algU$ from the inside. We remark that our notion of a $(\algL,\algU)$-clonoid is a strengthening of the original definition of an \emph{$\algL$-clonoid} in \cite{AM-clonoid}, which is a set of operations $\mathcal C \subseteq \bigcup_{n \in \N} L^{U^n}$, that is only closed under the \emph{term} operations of $\algL$. 

Studying a wreath product goes hand in hand with studying clonoids. In Observation \ref{observation:absorbing} we saw, for instance, that $0$-absorbing polynomials $p$ over a 2-nilpotent algebra $\algA = \algL \otimes \algU$ satisfy $\proj{U} p(\bar x) = 0$ and $\proj{L}(p(\bar x)) = \hat p(\proj{U}(\bar x))$ for some $\hat p \colon U^k \to L$. Thus $p$ can be identified with the map $\hat p \colon U^k \to L$; composing $p$ with polynomials $\Pol(\algA)$ from the inside, and those who preserve $L\times\{0\}$ from the outside corresponds to composing $\hat p$ with polynomials of $\algU$ from the inside, and polynomials of $\algL$ from the outside. Thus the closure under such compositions gives rise to a $(\algL,\algU)$-clonoid.

We are going to show that, for coprime elementary Abelian groups $\algL$ and $\algU$, a $(\algL,\algU)$-clonoid consists either of only constant functions, or is able to express conjunctions in the following sense:

\begin{proposition} \label{proposition:clonoid}
Let $\algL = (\Z_q^l,+,0,-)$ and $\algU = (\Z_p^k,+,0,-)$ for distinct primes $p \neq q$, and let $\mathcal C$ be a $(\algL,\algU)$-clonoid which contains a non-constant operation. Then there exist a hyperplane $H < \algU$ (i.e. a maximal subgroup), an element $l \in L \setminus \{0\}$, and $n$-ary functions $t_n \in \mathcal C$ for every $n \in \N$ such that
\begin{itemize}
\item $t_n(u_1,\ldots,u_n) = l$ if $u_i \in H$ for all $i=1,\ldots,n$
\item $t_n(u_1,\ldots,u_n) = 0$ else.
\end{itemize}
Furthermore $t_n$ can be constructed as a linear combination of expressions $t_1(\sum_{i=1}^n b_i \cdot u_i + c)$ in time $O(p^{n+1})$.
\end{proposition}

We remark that, for the special case $l = k = 1$, this was already shown in \cite[Lemma 5.2]{AM-zpq}. Further \cite[Lemma 3.1.]{IKK-examples}) contains a version for 1-dimensional vector spaces $\algL$, $\algU$ over arbitrary coprime fields. In order to prove Proposition \ref{proposition:clonoid}, we first prove the following lemma:

\begin{lemma} \label{lemma:pzero}
Let $\algL$, $\algU$ and $\mathcal C$ be as in Proposition \ref{proposition:clonoid}. Then there is a unary $t \in \mathcal C$, such that $t$ is not constant and
\begin{enumerate}
\item $t(0) = 0$
\item $\sum_{u \in U}t(u) \neq 0$
\item $t$ only depends on $\algU / H \cong \Z_p$, for a hyperplane $H < \algU$.
\item $t(u) = 0$ if $u \in H$ and $t(u) = l \neq 0$ else.
\end{enumerate}
\end{lemma}

\begin{proof}
By assumption $\mathcal C$ contains a non-constant function $t$. Without loss of generality we can assume that $t$ is unary, otherwise we substitute all but one of its variables by suitable constants. The additional properties can be shown as follows:

\begin{enumerate}
\item If $t$ does not have this property, we take $t'(x) = t(x)-t(0)$ instead.

\item If $\sum_{u \in U}t(u) = 0$ we take the map $t'(x) = t(x+a)-t(a)$ instead, where $a$ is an element with $t(a) \neq 0$. Then $\sum_{u \in U} t'(u) = - p^k \cdot t(a)$, which is not equal to $0$ in $\algL$, since $p \neq q$. Note that $t'$ is still not constant, and $t'(0) = 0$.

\item For every hyperplane $H < U$, let us define the map $t_H(x) = \sum_{h \in H} t(x + h)$. We claim that there is a hyperplane $H$ such that $t_H$ is not constant. For contradiction, assume that this is not the case, i.e. every $t_H$ is constant.

Let us define $c = \sum_{u \in U}t(u) = \sum_{u \in U \setminus \{0\}}t(u)$, which is not equal to $0$ by (2). For any hyperplane $H$, and any enumeration $a_1+H,\ldots,a_p+H$ of its cosets we have $t_H(a_1) + t_H(a_2) + \cdots + t_H(a_p) = \sum_{u \in U} t(u) = c$. Therefore, for every $H$ the constant function $t_H(x)$ must be equal to $p^{-1} c$ (where $p^{-1}$ is a multiplicative inverse of $p$ modulo $q$).

Next, observe that there are $\frac{p^k-1}{p-1} = \sum_{i = 0}^{k-1}p^i$ hyperplanes of $\Z_p^k$, so $\sum_{H} t_H(0) = p^{-1} c \cdot (\sum_{i = 0}^{k-1}p^i)$. On the other hand, every element $x \in \Z_p^k \setminus \{0\}$ is contained in $\frac{p^{k-1}-1}{p-1} = \sum_{i = 0}^{k-2}p^i$ distinct hyperplanes. %(by symmetry we only need to consider $x = (x_1,\ldots,x_k) = (1,0,\ldots,0) \in \Z_p^k$. This vector satisfies a linear equation $\sum a_i x_i = 0$ iff $a_1 = 0$. There are $p^{k-1}-1$ many non-trivial such equations, which gives us the total number of $\frac{p^{k-1}-1}{p-1}$ hyperplanes containing $x$).
Since $t(0) = 0$, this implies $\sum_{H} t_H(0) = (\sum_{x \in U \setminus \{0\}} t(x)) = c \cdot (\sum_{i = 0}^{k-2}p^i)$.
These two equations imply that $0 = p^{-1}c$, which is a contradiction to $c \neq 0$. Thus there is a hyperplane $H$, such that $t_H$ is not constant. Without loss of generality we can further assume that (1) and (2) hold for $t_H$ (otherwise we repeat the constructions in (1) and (2)).

\item Assume that $t$ satisfies items (1),(2),(3). It follows from (2) that $l:= t(0) + t(a) + t(2\cdot a) + \cdots t((p-1)\cdot a) \neq 0$, where $0, a, 2\cdot a,\ldots, (p-1)\cdot a$ are representatives of the cosets of $H$. Then the map $\sum_{i = 0}^{p-1} t(i \cdot x)$ maps $H$ to $0$ and its complement to $l$.
\end{enumerate}
\end{proof}

Proposition \ref{proposition:clonoid} now follows straightforward from the existence of this unary function:

\begin{proof}[Proof of Proposition \ref{proposition:clonoid}]

Let $t \in \mathcal C$ be the unary function constructed in Lemma \ref{lemma:pzero}. Then we set $t_1(u) = t(u)$. Note that $t_1$ only depends on $\algU/H \cong \Z_p$ and it's image is $\{0,l\}$, which generates a subgroup of $L$, isomorphic to $\Z_q$. It follows from the analysis of $(\Z_q,\Z_p)$-clonoids in \cite{AM-zpq} that we can obtain all function $t_n$ by the recursion

$$t_{n+1}(u_1,\ldots,u_n,u_{n+1}) = p^{-1} \cdot (\sum_{i = 0}^{p-1} t_{n}(u_1,\ldots,u_{n-1}, u_n - i \cdot u_{n+1}) - \sum_{i = 1}^{p-1} t_{n}(u_1,\ldots u_{n-1},u_n - i\cdot a)),$$

where $0, a, 2a,\ldots, (p-1)a$ are transversal to the cosets of $H$, and $p^{-1}$ is an integer such that $p^{-1} \cdot p = 1 \mod q$. Note that by this recursion $t_{n}$ can be written as a linear combination of the $p^{n+1}$ many expressions $t_1(\sum_{i=1}^n \alpha_i u_i + c)$.
\end{proof}

As a direct consequence of Proposition \ref{proposition:clonoid}, $\mathcal C$ contains \emph{all} operations that only depend on $\algU/H$ and map to the subgroup generated by $l$:

\begin{corollary} \label{corollary:clonoid}
Let $\algL = (\Z_q^l,+,0,-)$ and $\algU = (\Z_p^k,+,0,-)$ with $q \neq p$, and let $\mathcal C$ be a $(\algL,\algU)$-clonoid which contains a non-constant operation. Then there exist a hyperplane $H < \algU$ and a cyclic subgroup $\langle l \rangle \leq \algL$, such that for every operation $f \colon (U/H)^n \to \langle l \rangle$, there is a $f' \in \mathcal C$ with $f'(x_1,\ldots,x_n) = f(x_1H, \ldots, x_nH)$, which can be constructed as linear combination of expressions $t_1(\sum_{i=1}^n b_i \cdot u_i + c)$ in time $O(p^{2n+1})$.
\end{corollary}

\begin{proof}
Let $H$, $l$, and $t_n$, $n \geq 1$ be as in Proposition \ref{proposition:clonoid}. Note that the function $f'$ can be obtained as a linear combinations of the form $f'(\bar x) = \sum_{\bar a \in (U/H)^n} f(\bar a) \cdot t_n(\bar x- \bar a')$, where $\bar a' = (a_1',\ldots,a_n')$ is a tuple, such that $a_i = a'_i H$ for every $i$. Expressing all of the $p^n$-many summands $t_n(\bar x- \bar a')$ as a linear combination of expressions $t_1(\sum_{i=1}^n b_i \cdot u_i + c)$ according to Proposition \ref{proposition:clonoid}, allows us to compute  such a linear combination also for $f$ in exponential time $O(p^{2n+1})$.
\end{proof}

\section{Absorbing Polynomials} \label{sect:absorbing}
In this section we discuss, based on Theorem \ref{theorem:absorbing}, how absorbing polynomials can be used to characterize the higher commutator in nilpotent Mal'tsev algebras $\algA$. In particular, we show that in the finite case any block of the congruence $[\alpha_1,\ldots,\alpha_k]$ can be written as the image of the corresponding blocks of $\alpha_1,\ldots,\alpha_k$ under some polynomial $p \in \Pol(\algA)$.

\begin{lemma} \label{lemma:absorbing}
Let $\algA$ be a nilpotent Mal'tsev algebra, $0 \in A$, $\alpha_1,\ldots,\alpha_k \in \Con(\algA)$, $\gamma = [\alpha_1,\ldots,\alpha_k]$ and $L =  [0]_{\alpha_1} \times [0]_{\alpha_2} \times \cdots \times [0]_{\alpha_k}$. Then
\begin{enumerate}
\item[(1)] $\gamma$ is generated, as a congruence, by the pairs
$$R = \{ (0,p(\bar b)) \colon \bar b \in L \text{ and } p \in \Pol^k(\algA) \text{ is $0$-absorbing} \}$$
\item[(1')] $\gamma$ is generated, as a congruence, by the pairs
$$R' = \{ (0,p(\bar b)) \colon \bar b \in L \text{ and } p \in \Pol^k(\algA) \text{ such that } p|_L \text{ is $0$-absorbing} \}$$
\item[(2)] If $|A|$ is finite, then $[0]_\gamma = p_1(L) + p_2(L) + \cdots + p_m(L)$, where all $p_i$ are $k$-ary $0$-absorbing polynomials (and the loop product with respect to $+$ is left associated).
\end{enumerate}
\end{lemma}

\begin{proof}
By Theorem \ref{theorem:absorbing} the commutator $\gamma = [\alpha_1,\ldots,\alpha_k]$ is generated by the pairs 
$$S = \{(c(\bar a),c(\bar b)) \colon (a_i, b_i) \in \alpha_i \text{ for all } i, \text{ and } c \in \Pol^k(\algA) \text { is absorbing at } \bar a \}.$$
It is clear that $R \subseteq S \subseteq \gamma$. For the other inclusion, it is enough to prove that a congruence generated by $R$ contains $S$. To see this, let $(c(\bar a),c(\bar b)) \in S$ for a polynomial $c$ that is absorbing at $\bar a$. We then define the polynomial $p(x_1,\ldots,x_n) = c(x + a_1,\ldots, x + a_k) / c(a_1,\ldots,a_k)$, and the tuple $\bar u = (b_1 / a_1,\ldots, b_k/a_k) \in L$. Note that $p$ is $0$-absorbing, and thus $(0,p(\bar u)) \in R$. Since $(0,p(\bar u)) + (c(\bar a),c(\bar a)) = (c(\bar a),c(\bar b))$, every congruence that contain $(0,p(\bar u))$, also must contain $(c(\bar a),c(\bar b))$. This finishes the proof of (1).

For (1') note that $R \subseteq R'$, hence the congruence generated by $R'$ must contain $\gamma$. On the other hand, it is not hard to see that $R' \subseteq \gamma$, by applying the Definition \ref{definition:commutator} for $\gamma$ and any $p \in \Pol(\algA)$, such that $p|_L$ is $0$-absorbing, 

For (2) note that, since $\gamma$ is a congruence, $[0]_\gamma$ is closed under the loop multiplication $x+ y$. Therefore, for every list of $k$-ary $0$-absorbing polynomials $p_1,\ldots,p_m$ we have $p_1(L) + p_2(L) +\cdots + p_m(L) \subseteq [0]_\gamma$. Let $N$ be a maximal set of the form $p_1(L) + p_2(L) + \cdots + p_m(L)$ (such a maximal set exists, since $A$ is finite).

We claim that $N = [0]_\gamma$. For contradiction assume that there is an element $d \in [0]_\gamma \setminus N$. By Mal'tsev's chain lemma (see e.g. \cite[Lemma 3.1.(1)]{IK-CSAT}), we know that there is a chain of elements $0 = a_0, a_1, \ldots, a_n, a_{n+1} = d$, such that for every $i=0,\ldots,n$, either $(a_i,a_{i+1})$ or $(a_{i+1},a_i)$ is the image of a pair $(0,p(b_1,\ldots,b_k)) \in R$ under a unary polynomial $u \in \Pol(\algA)$. Without loss of generality we can further assume that $a_i \in N$ for all $i = 1,\ldots,n$ (otherwise we replace $d$ by the minimal $a_i$, which is not in $N$).

If $(a_{n},a_{n+1}) = (u(0),u(p(b_1,\ldots,b_k)))$, then $a_{n} \backslash a_{n+1} = u(0) \backslash u(p(b_1,\ldots,b_k))$ is the image of $L$ under the $0$-absorbing polynomial $q(x_1,\ldots,x_n) = u(0)\backslash u(p(x_1,\ldots,x_k))$. But then $d = a_{n} + (a_n \backslash a_{n+1}) \in N + q(L)$, which is a contradiction to the maximality of $N$. Symmetrically, if $a_{n} \backslash a_{n+1} = u(p(b_1,\ldots,b_k)) \backslash u(0)$ is in the image of $L$ under the $0$-absorbing polynomial $q(x_1,\ldots,x_k) = u(p(b_1,\ldots,b_k)) \backslash u(0)$, then $d = a_{n} + (a_n \backslash a_{n+1}) \in N + q(L)$, which is a contradiction.
\end{proof}

In the proof of Lemma \ref{lemma:verbalminimal} we only used that $\Pol(\algA)$ contains loop operations $+,/,\backslash$ as polynomial, thus the lemma also holds in this more general setting. Moreover, Lemma \ref{lemma:verbalminimal}(2) implies that the congruence blocks of all congruences $\gamma$ that can be defined using $\mathbf 1_A$ and (nested) commutator brackets (i.e. $\gamma \in \Xi(\mathbf 1_A)$ in the notation of \cite{wires-supernilpotent}), are the image of some polynomial operation.

\section{Wreath product representations} \label{sect:wreathproduct}
In this section we discuss some properties of the wreath product representation of nilpotent algebras from Definition \ref{definition:wreathproduct}. In particular we investigate the behaviour of such wreath product representation with respect to supernilpotent congruences, in order to prove Proposition \ref{proposition:main}.

\subsection{The lower central series}
By Lemma \ref{lemma:absorbing} we know that, for each element of the series $\alpha \geq [\mathbf 1_A,\alpha] \geq [\mathbf 1_A,[\mathbf 1_A,\alpha]] \geq \cdots$ for $\alpha \in \Con(\algA)$ in a finite nilpotent Mal'tsev algebra, its congruence blocks can be defined as the image of $[0]_\alpha$ and $A$ under a suitable polynomial $h(\bar x)$. Using the wreath product representation, we can prove the existence of such $h$ with even nicer properties:

\begin{lemma} \label{lemma:verbalminimal}
For a finite nilpotent Mal'tsev algebra $\algA$ and $\alpha \in \Con(\algA)$, let $\mathbf 0_A = \alpha_0 < \alpha_{1} < \ldots < \alpha_n = \alpha$ be defined by $\alpha_{i} = [\mathbf 1_A,\alpha_{i+1}]$ and let $\algA = \bigotimes_{i = 1}^n \algL_i \otimes \algA/\alpha$ the corresponding wreath product representation. Then there is a $k \in \N$ and a polynomial $h \in \Pol^k(\algA)$, such that $h(0,\ldots,0) = 0$, and for $\bar x \in ([0]_\alpha)^k$:
\begin{itemize}
\item $\proj{L_j}  (h(\bar x)) = 0$ for $j > 1$,
\item there is a surjective $\hat h \colon L_n^k \to L_1$, such that $\proj{L_1}  (h(\bar x)) = \hat h(\proj{L_n}  (\bar x))$.
\end{itemize}
\end{lemma}

\begin{proof}
We prove the statement by induction on $n$. If $n = 1$, $\alpha = \alpha_1$ is central, and we simply set $h(x)=x$. For an induction step $(n-1) \to n$, by the induction assumption, there is already a polynomial $h'$, such that for $\bar x$ from $[0]_\alpha$:

\begin{itemize}
\item $\proj{L_2}  (h'(\bar x)) = \hat h'(\proj{L_n}  (\bar x))$, such that $\hat h' \colon L_n^k \to L_2$ is surjective, and $\hat h'(0,\ldots,0) = 0$
\item $\proj{L_j}  (h'(\bar x)) = 0$ for every $j > 2$.
\end{itemize}

By Lemma \ref{lemma:absorbing} (1), $\alpha_{1} = [\mathbf 1_A,\alpha_{2}]$ is generated by pairs $(0,c(a,x))$, such that $c \in \Pol(\algA)$ is absorbing, $a \in A$ and $x \in [0]_{\alpha_2}$. For every such $c$ and $a$ let us define $c_a(x) = c(a,x)$. With respect to the wreath product representation $\bigotimes_{i = 1}^n \algL_i \otimes \algA/\alpha$ this implies that
$$c_a((l_1,l_2,0,\ldots,0)) = (\hat c_a(l_2),0, \ldots, 0),$$

for some non-constant $0$-absorbing function $\hat c_a \colon L_2 \to L_1$ (cf. Observation \ref{observation:absorbing}). As in Lemma \ref{lemma:absorbing} (2), one can prove that every element of $[0]_{\alpha_1} = L_1 \times \{0\} \times \cdots \times \{0\}$ is in the closure of elements $(\hat c_a(l_2),0, \ldots, 0)$ under the loop product $+$. Thus there is a polynomial $p(\bar x) = c_{1,a_1}(x_1) + c_{2,a_2}(x_2) +\cdots + c_{r,a_r}(x_r)$ that surjectively maps $[0]_{\alpha_2}$ to $[0]_{\alpha_1}$, and only depends on $L_2$. The composition $h = p(h'(\bar y_1), h'(\bar y_2),\ldots, h'(\bar y_s))$ (such that all appearing variables are distinct) is then a polynomial with the desired properties.
\end{proof}

\subsection{Supernilpotent congruences}
In this section, we give a characterization of supernilpotent congruences in (elementary) wreath products.  We first prove the following lemma:

\begin{lemma} \label{lemma:2factors}
Let $\algA = \algL_1 \otimes \algL_2 \otimes \algA / \alpha_2$ be a finite nilpotent Mal'tsev algebra, such that $\exp(\algL_1) \neq \exp(\algL_2)$ are two distinct primes. Then, either
\begin{itemize}
\item $[\mathbf 1_A,\alpha_2] = \mathbf 0_A$, and for every $p \in \Pol(\algA)$,  $\proj{L_1}  p$ does not depends on $L_2$, or
\item $\alpha_2$ is not supernilpotent, and there is a $p \in \Pol(\algA)$, such that $p((l_1,l_2,0)) = (\hat p(l_2), 0, 0)$, for a non-constant $\hat p \colon L_2 \to L_1$ 
\end{itemize}
\end{lemma}

\begin{proof}
We distinguish two cases. If  $[\mathbf 1_A,\alpha_2] = \mathbf 0_A$, then we know by Theorem \ref{theorem:freese-mckenzie} that there is an affine algebra $\algL$, such that $\algA = \algL_1 \otimes \algL_2 \otimes \algA / \alpha_2 \cong \algL \otimes \algA / \alpha_2$. Without loss of generality, we can identify the universe of $\algL$ with $L_1 \times L_2$. Note that $\algL$ is polynomially equivalent to a module, and has $\algL_2$ as a quotient. So (by the fundamental theorem of finitely generated modules) $\algL \cong \algL_1 \times \algL_2$. In particular this implies that there is no polynomial $p \in \Pol(\algA)$, such that that $\proj{L_1} p$ depends on $L_2$.

%need all the versions for the absorbing lemma (locally absorbing on L etc)
In the remaining case $\mathbf 0_A < [\mathbf 1_A,\alpha_2] \leq \alpha_1$ holds. By Lemma \ref{lemma:verbalminimal}, there is a polynomial $h$ such that, for $\bar x$ from $[0]_{\alpha_2}$
$$h(\bar x) = (\hat h(\proj{L_2}(\bar x)),0, 0),$$
and $\hat h$ is not constant.

If we close $h$ under compositions with the loop multiplication $+$ (and constants $(0,l_2,0)$ from the inside, respectively constants $(l_1,0,0)$ from the outside), we obtain polynomials $f \in \Pol(\algA)$ with
$$f(\bar x) = (\hat f( \proj{L_2} (\bar x)), 0, 0 ), \text{ for } \bar x \in L_1 \times L_2 \times \{0\},$$
 for every $\hat f \colon L_2^n \to L_1$ that is in $((L_1,+,0,-),(L_2,+,0,-))$-clonoid generated by $\hat h$. In particular, by Corollary \ref{corollary:clonoid}, there are polynomials of arbitrary arity, which are $0$-absorbing on $[0]_{\alpha_2} = L_1 \times L_2 \times \{0\}$, but not constant. By Lemma \ref{lemma:absorbing} (1'), $\alpha_2$ is not supernilpotent.
\end{proof}

We next use Lemma \ref{lemma:2factors} to characterize supernilpotent congruences by their (elementary) wreath product representations. It can be seen as a special case of Theorem \ref{theorem:mayrdecomposition}.

\begin{proposition} \label{proposition:supernilpotentwreath}
Let $\algA$ be a finite nilpotent Mal'tsev algebra, and let $\algA = \bigotimes_{i = 1}^k \algL_i \otimes \algA / \alpha$ be an elementary wreath product representation. Then $\alpha$ is supernilpotent if and only if for every polynomial $f \in \Pol(\algA)$, 
$\proj{L_i}  f$ only depends on $L_j$ if $i \leq j$ and $\exp(\algL_i) = \exp(\algL_j)$.
\end{proposition}

\begin{proof}
We start by proving the 'only if' direction. For contradiction, assume that there is a finite nilpotent Mal'tsev algebra $\algA = \bigotimes_{i = 1}^k \algL_i \otimes \algA / \alpha$, such that $\alpha$ is supernilpotent, but there are components $L_i$ and $L_j$ of different exponent, and a polynomial $f$ such that $\proj{L_i} f$ depends on $L_j$. Without loss of generality we can assume that $i=1$ is the only value of $i$, for which this is true; otherwise we take the quotient $\algA/ \alpha_{i-1}$, for the maximal such $i$.

We can further assume that $j=2$: if this is not the case, let $j$ be the minimal index, such that $L_1$ and $L_j$ have different exponent, and there is a polynomial $f$ such that $\proj{L_1}  f$ depends on $L_j$. Note that, by the minimality of $j$, for any polynomial $f \in \Pol(\algA)$, its projection $\proj{L_1}  f$ only depends on $L_m$ with $m < j$, if $m$ has the same exponent as $L_1$. Let $I = \{ 1 < m < j \colon \exp(\algL_1) \neq \exp(\algL_m) \}$; then, $\prod_{m \notin I} \algL_i \times \algA/\alpha$ is a well-defined nilpotent algebra (cf. Observation \ref{observation:kernel}), and also a counterexample to the theorem. So, without loss of generality, we can assume that all factors $\algL_m$ for $m < j$ have the same exponent as $\algL_1$. Without loss of generality we can further assume $j = 2$ (if this is not the case, we can move $\algL_j$ to the second position; this is possible since, by the uniqueness of $i=1$, no component other than $\algL_1$ depends on $\algL_j$). By Lemma \ref{lemma:2factors}, $\alpha_j = \alpha_2$ cannot be supernilpotent. Since $\alpha_2 \leq \alpha$, $\alpha$ is not supernilpotent - contradiction.

For the 'if' direction, note then that by Observation \ref{observation:kernel}, for every supernilpotent congruence $\alpha$, and every prime $p$, the kernel of the projections to $\algA/\alpha$ and all components $L_i$ with $\exp(\algL_i) \neq p$ is a congruence $\alpha_{p} \leq \alpha$, that has blocks of some $p$-power size. It is further easy to see that, for every $p$, $\alpha$ is the composition of $\alpha_p$ and $\bigwedge_{q \neq p} \alpha_q$.  Thus $\alpha$ satisfies the characterization of supernilpotent congruences in Theorem \ref{theorem:mayrdecomposition}.
\end{proof}

As a direct corollary we get a generalization for the case in which the affine components $\algL_i$ have prime power exponents:

\begin{corollary} \label{corollary:supernilpotentwreath}
Let $\algA$ be a finite nilpotent Mal'tsev algebra, and let $\algA = \bigotimes_{i = 1}^k \algL_i \otimes \algA / \alpha$ such that for every $i$, $\exp(\algL_i) = p_i^{m_i}$, for some prime $p_i$ and $m_i \geq 1$. Then $\alpha$ is supernilpotent if and only if for every polynomial $f \in \Pol(\algA)$, 
$\proj{L_i}  f$ only depends on $L_j$ if $i \leq j$ and $p_i = p_j$.
\end{corollary}

\begin{proof}
This follows straightforwardly from refining $\bigotimes_{i = 1}^k \algL_i \otimes \algA / \alpha$ into an elementary wreath product representation, as in Observation \ref{observation:subdivision}, and applying Proposition \ref{proposition:supernilpotentwreath}.
\end{proof}

We next prove that, for supernilpotent congruences $\alpha \in \Con(\algA)$ it is possible to separate the affine components of different prime expontents by unary polynomials. For this we first need the following lemma about the loop operation in nilpotent algebras of prime power size:

\begin{lemma} \label{lemma:loops}
Let $p$ be a prime and $\algA = \bigotimes_{i = 1}^k \algL_i$ be a nilpotent Mal'tsev algebra such that $\exp(\algL_i) = p$ for every $i$. For any prime $q$, let us define the unary polynomials $u_{q}(x) = q \cdot x = x+ \cdots + x$ (where this loop product is left associated) and $u_{q^{l+1}}(x) = u_q(u_{q^l}(x))$ for every $l \geq 1$. Then
\begin{enumerate}
\item $u_{q^l}(x)$ is bijective on $A$ for all $q \neq p$ and $l \geq 1$,
\item $u_{p^l}(x) \approx 0$ for all $l \geq k$.
\end{enumerate}
\end{lemma}

\begin{proof}
We prove both statements by induction on $k=1$. For $k=1$, $(A,+,0,-) \cong \Z_p^m$ for some $m$, so the Lemma clearly holds. 

To show (1) in general, note that by induction hypothesis we know that the projection of $u_{q^l}$ to $\prod_{i = 2}^k L_i$ only depends on $\prod_{i = 2}^k L_i$ (in a bijective manner). Furthermore $\proj{L_1}  u_{q^l}(x) = (u_{q^l}(\proj{L_1} (x)) + \hat u_{q^l}(\proj{L_2 \times \cdots \times L_k}(x))$, for some map $\hat u_{q^l}$. Since $p \neq q$, the affine part $u_{q^l}(\proj{L_1}(x)) = q^l \cdot \proj{L_1}(x)$ is a bijective polynomial of $\algL_1$. Thus $u_{q^l}$ is bijective on all of $A$.

In order to see (2), note that by the induction hypothesis $\proj{L_j} u_{p^{l-1}}(x) = 0$ for all $j > 1$ and $ \proj{L_1}(u_{p^{l-1}}(x)) = \hat u_{p^{l-1}}(\proj{L_2 \times \cdots \times L_k}(x)))$. Since $\algL_1$ has exponent $p$ it follows that $\proj{L_1}u_{p^{l}}(x) = p \cdot \hat u_{p^{l-1}}(\proj{L_2 \times \cdots \times L_k}(x)) = 0$, and thus $u_{p^{l}}(x) \approx 0$.
\end{proof}

As a direct consequence of Lemma \ref{lemma:loops} we obtain:

\begin{lemma} \label{lemma:loops2}
Let $\algA$ be a finite nilpotent Mal'tsev algebra, $\alpha \in \Con(\algA)$ be a supernilpotent congruence. Further let $\bigotimes_{i=1}^n \algL_i \otimes \algA/\alpha$ be an elementary wreath product representation. Then, for every prime $p$, there is a polynomial $r_p \in \Pol^1(\algA)$, such that for every $x \in [0]_\alpha$:
$$\proj{L_j}(r_p(x)) = \begin{cases}
 \proj{L_j}(x) &\text{ if } \exp(\algL_j) = p,\\
0 &\text{ else.}\end{cases}$$
\end{lemma}

\begin{proof}
Note that $([0]_\alpha, +,/,\backslash) = (\prod_{i=1}^n L_i \times \{0/_\alpha \}, +,/,\backslash)$ is a supernilpotent loop. By Proposition \ref{proposition:supernilpotentwreath} it is the direct product of loops of the form $K_p = \{ x \colon \proj{L_i}(x) = 0 $ if $\exp(\algL_i) \neq p \}$, where $p$ runs through all prime exponents $p$ of the affine algebras $\algL_i$. By Lemma \ref{lemma:loops} the polynomial $u_{q^n}$ is bijective on $K_p$ if $q \neq p$, and $0$ otherwise.

Let $t_p$ to be a composition (in any order) of all function $u_{q^n}$, such that $q$ divides $[0]_\alpha$ and $q\neq p$. Then $t_p$ is a bijection on $K_p$, and $0$ on all other $K_q$ for $q\neq p$.  We define $r_p$ to be a suitable power of $t_p$ that is the identity on $K_p$.
\end{proof}

\subsection{Main lemmas}
In this section we are going to construct certain circuits, that we use in the proof of Theorem \ref{theorem:main}.

\begin{lemma} \label{lemma:rho}
Let $\algA$ be a finite nilpotent Mal'tsev algebra, let $\lambda$ be its Fitting congruence, and let $\rho > \lambda$ such that $\algA / \lambda = \algM \otimes \algA / \rho$ for an elementary Abelian $\algM$. Then there is an elementary wreath product representation $\algA = \bigotimes_{i = 1}^j \algL_i \otimes \alg{M} \otimes \algA / \rho$, such that
\begin{itemize}
\item $\exp(\algL_1) \neq \exp(\algM)$
\item There is a polynomial operation $p \in \Pol^k(\algA)$, such that for $\bar x \in [0]_\rho^k$:
\begin{align*}
\proj{M} p(\bar x) = 0, \, \proj{L_j} p(\bar x) &= 0, \text{ for } j \neq 1, \\
\text{and } \proj{L_1} p(\bar x) &= \hat p(\proj{M}(\bar x))
\end{align*}
for a surjective operation $\hat p \colon M^k \to L_1$.
\end{itemize}
\end{lemma} 

\begin{proof}

Let $\mathbf 0_A = \rho_0 < \rho_1 < \cdots < \rho_{n-1} \leq \lambda < \rho_n = \rho$ be the series such that $\rho_{i-1} = [\mathbf 1_A,\rho_i]$. (Note that $\rho_{n-1} = [\mathbf 1_A,\rho] \leq \lambda$ follows from $\algA / \lambda = \algM \otimes \algA / \rho$). Let $\algA = \bigotimes_{i = 1}^{n} \algK_i \otimes \algM \otimes \algA / \rho$ be the corresponding wreath product representation.

Furthermore, let $\algK_i = \prod_{m=1}^{n_i} \algK_{i,q_m}$ a direct decomposition into factors whose exponents are powers of distinct primes $q_m$ (note that substituting every $\algK_i$ by its direct decomposition also results in a wreath product representation of $\algA$). By our assumptions $\rho$ is not supernilpotent, but $\lambda$ is. So, by Corollary \ref{corollary:supernilpotentwreath}, there must be an index $i$ with a component $\algK_{i,q}$, with $q \neq \exp(\algM)$. Let us pick the minimal such index $i$. By Lemma \ref{lemma:verbalminimal}, there is a polynomial $h \in \Pol(\algA)$, such that, restricted to $[0]_\rho$, $\proj{K_i} h$ is surjective and only depends on $M$.

Let $r_q$ be the unary polynomial with respect to the prime $q$ (and the supernilpotent $\lambda$) given by Lemma \ref{lemma:loops2}, and let $t = r_q \circ h$. Note that then $\proj{K_{i,q}} t(\bar x) = \hat t(\proj{M}(\bar x))$, for a surjective $\hat t \colon M^n \to K_{i,q}$, and $\proj{K_{j,p}} t(\bar x) = 0$ for all other components.

By the minimality of $i$, we know that $p\neq q$ for all components $\algK_{m,p}$ with $m<i$. By Corollary \ref{corollary:supernilpotentwreath} (and the supernilpotence of $\lambda$), $\proj{K_{m,p}} f$ does not depends on $\algK_{i,q}$ for any polynomial $f$. Thus, we can obtain another wreath product representation of $\algA$, by moving $\algK_{i,q}$ to the first position. If $\exp(\algK_{i,q}) = q$, we define $\algL_1 = \algK_{i,q}$, and $p(\bar x) = t(\bar x)$. If $\exp(\algK_{i,q}) = q^{s}$, then we define $p(\bar x) = q^{s-1} \cdot t(\bar x)$, and $\algL_1 = \{ q^{s-1} \cdot x : x \in \algK_{i,q} \}$. The statement of the lemma then follows after refining to an elementary wreath product representation $\algA = \bigotimes_{i = 1}^j \algL_i \otimes \alg{M} \otimes \algA / \rho$ (see Observation \ref{observation:subdivision}).
\end{proof}

Together with Corollary \ref{corollary:clonoid}, this implies the following:

\begin{proposition} \label{proposition:main}
Let $\algA$ be a finite nilpotent Mal'tsev algebra, let $\lambda$ be its Fitting congruence, and let $\rho > \lambda$ such that $\algA / \lambda = \algM \otimes \algA / \rho$ for an elementary Abelian $\algM$. Then there is an elementary wreath product representation $\algA = \bigotimes_{i = 1}^j \algL_i \otimes \alg{M} \otimes \algA / \rho$, and a hyperplane $H < (M,+,0,-)$ such that
\begin{itemize}
\item $\exp(\algL_1) \neq \exp(\algM)$
\item The module equivalent to $\algL_1$ is cyclic, i.e. it has a single generator $l$
\item For every function $\hat f \colon M^n \to L_1$ that only depends on $M/H$, there is a circuit $f$ over $\algA$ such that for $\bar x \in [0]_\rho^n$:
\begin{align*}
\proj{L_j} f(x_1,\ldots,x_n) &= 0, \text{ for } j \neq 1, \\
\proj{L_1} f(x_1,\ldots,x_n) &= \hat f(\proj{M}(\bar x)).
\end{align*}
Further, $f$ can be constructed in time $2^{O(n)}$.
\end{itemize}
\end{proposition}

\begin{proof}
By Lemma \ref{lemma:rho}, we know that there is a wreath product representation $\algA = \bigotimes_{i = 1}^j \algL_i \otimes \alg{M} \otimes \algA / \rho$, and a polynomial $h$ such that for $\bar x$ in $[0]_\rho$, $\proj{L_j} h(\bar x)$ for $j \neq 1$, and $\proj{L_1} h(\bar x) = \hat h(\proj{M}(\bar x))$, for surjective $\hat h$. In particular $\hat h$ is not constant. By closing $h$ under the loop product $+$ and constants from $[0]_\rho$ from the inside, and $+$ and constants from $L_1$ from the outside, we obtain a polynomial $f \in \Pol(\algA)$ for every function $\hat f \colon M^n \to L_1$ that is in the $((L_1,+,0,-),(M,+,0,-))$-clonoid generated by $\hat h$, such that 
$$f(\bar x) = \proj{L_1} f(x_1,\ldots,x_n) = \hat f(\proj{M}(\bar x)).$$
By Corollary \ref{corollary:clonoid}, we obtain such an $f$ for all function $\hat f$ that only depend on $M/H$ for some hyperplane $H < M$ and map to cyclic subgroup $\langle l \rangle$. Note that the unary polynomials of $\algA$ that preserve $0$ are exactly the operations of the form $(l_1,0,\ldots,0) \mapsto (r \cdot l_1,0,\ldots,0)$, when restricted to $L_1\times \{0\} \times \cdots \times \{0\}$ (where $r \cdot l_1$ is a scalar multiplication in the module associated to $\algL_1$). Thus, if we further close by compositions with all unary polynomials $g$ of $\algA$ that preserve $0$, we can even obtain all maps $\hat f$ whose image is in the \emph{submodule} generated by $l$. If this submodule is equal to $\algL_1$, we are done. Otherwise we apply Observation \ref{observation:subdivision}, to refine the wreath product accordingly. Note that, as in Corollary \ref{corollary:clonoid} $\hat f$ has a representation as a sum of expressions $g \circ \hat t_1(\sum_{i=1}^n b_i \cdot u_i + c)$ that can be computed in exponential time $2^{O(n)}$. If we substitute all occurrences of $+$ by circuits defining the loop multiplication $+$, and unary operations $\hat t_1$ and $g$ by polynomials defining them, we are done.
\end{proof}

\section{Proof of Theorem \ref{theorem:main}} \label{sect:proof}

In this section we complete the proof of Theorem \ref{theorem:main}. Using Proposition \ref{proposition:main} we construct suitable subexponential reductions from 3-SAT, respectively $C$-COLOR, i.e. the problem of coloring graphs with $C$-colors, to $\CSAT(\algA)$ and $\CEQV(\algA)$. This, together with the following two consequences of the exponential time hypothesis (ETH) then implies Theorem \ref{theorem:main}.

\begin{theorem} \label{theorem:ETC}
Assume that the exponential time hypothesis is true. Then 
\begin{enumerate}
\item $3$-SAT cannot be solved in time $2^{o(n+m)}$, where $n$ is the number of variables, and $m$ is the number of clauses in the input formula.
\item For any $C>2$, $C$-COLOR cannot be solved in time $2^{o(|V|+|E|)}$, where $V$ is the vertex set, and $E$ the edge set of an input graph.
\end{enumerate}
\end{theorem}

We recall that the original definition of the exponential time hypothesis \cite{IP-ETH} is the statement that there is a $c>0$ such that $3$-SAT cannot be solved by any algorithm of complexity $2^{cn + o(n)}$. Statement (1) follows from ETH by the so called Sparsification Lemma in \cite{IP-ETH} (see e.g. Theorem 14.4 in \cite{CFKLLMPPS-parameterized-algorithms}). Statement (2) follows directly from (1) and the existence of linear reductions from $3$-SAT to $C$-COLOR (cf. Theorem 14.6 in \cite{CFKLLMPPS-parameterized-algorithms}).

\begin{proof}[Proof of Theorem \ref{theorem:main}]
Let $\algA$ be a nilpotent Mal'tsev algebra of Fitting length $k \geq 2$. By Lemma \ref{lemma:upperF} it has an upper Fitting series $\mathbf 0_A = \lambda_0 < \lambda_1 < \cdots < \lambda_{k-1} < \lambda_k = \mathbf 1_A$ of length $k$. Further let $0 \in A$ be an arbitrary constant. 

First, note that there is a congruence $\rho > \lambda_{k-1}$, such that $\algA / \lambda_{k-1} = \algM \otimes (\algA/ \rho)$ and there is a polynomial $h(\bar x) \in \Pol(\algA)$, which maps $A/\lambda_{k-1}$ to $[0]_{\rho/\lambda_{k-1}} = M \times \{0\}$. This follows directly from Lemma \ref{lemma:verbalminimal}. Furthermore, we can assume that $\exp(\algM) = p$ for a prime $p$ (if this is not the case, we apply Observation \ref{observation:subdivision} to $\algM$).

We are first going to prove the theorem for the case in which $p$ is an odd prime. Then, the idea is to reduce p-COLOR, i.e. the p-coloring problem for graphs in subexponential time to $\CSAT(\algA)$ (or $\CEQV(\algA)$ respectively). To construct such a reduction, we first need to prove the following claim: \\
 
\textbf{Claim 1:} There is an elementary wreath product representation $\algA = \bigotimes_{i = 1}^m \algL_i \otimes \algM \otimes (\algA/ \rho)$, a hyperplane $H < (M,+,0,-)$ and an element $l \in L_1$, such that for every $n \in \N$, there are circuits $t_n(\bar x)$ of arity $n^{k-1}$, and size $2^{O(n)}$ such that for $\bar x \in ([0]_\rho)^{n^{k-1}}$:
\begin{itemize}
\item $\proj{M} (t_n(\bar x)) = 0$, $\proj{L_j} (t_n(\bar x)) = 0$ for all $1 < j \leq m$ and 
\item $\proj{L_1} (t_n(\bar x)) = \begin{cases} 0, &\text{ if } \exists i: \proj{M}(x_i) \in H \\
l \neq 0, &\text{ else.}\end{cases}$\\
\end{itemize}

\begin{proof}[Proof of Claim 1:] We prove the statement by induction on $k$. If $k = 2$, this straightforwardly follows from Proposition \ref{proposition:main}. For an induction step $k-1 \to k$, note that by the induction hypothesis, there is already a wreath product representation of $\algA/\lambda_1 = \bigotimes_{i = 1}^r \algK_i \otimes \algM \otimes \algA/\rho$, an element $a \in K_1$ and circuits $s_n(x_1,x_2,\ldots,x_{n^{k-2}})$ of size $2^{O(n)}$, such that, when restricted to $\bar x \in [0]_{\lambda_1}^{n^{k-2}}$:
\begin{itemize}
\item $\proj{K_j} (s_n(\bar x)) = 0$ for all $1 < j \leq r$ and 
\item $\proj{K_1} (s_n(\bar x)) = \begin{cases} 0 &\text{ if } \exists i: \proj{M}(x_i) \in H \\
a &\text{ else}.\end{cases}$
\end{itemize}
Without loss of generality we can assume that the module associated with $\algK_1$ is cyclic, and generated by $a$ (otherwise we refine the wreath product accordingly, cf. Observation \ref{observation:subdivision}). Let $\rho_1 > \lambda_1$ be the congruence such that $\algA/\rho_1 = \bigotimes_{i = 2}^r \algK_i \otimes \algM \otimes \algA/\rho$. By Proposition \ref{proposition:main}, there is an elementary wreath product representation $\algA = \bigotimes_{j =1}^m \algL_j \otimes \bigotimes_{i = 1}^r \algK_i \otimes \algM \otimes \algA/\rho$, and there are circuits $f_n(x_1,x_2,\ldots,x_{n})$ of length $2^{O(n)}$, and a hyperplane $H_1 < (K_1,+,0,-)$ such that for $\bar x \in [0]_{\rho_1}^n$
\begin{itemize}
\item $\proj{L_j} (f_n(\bar x)) = 0$ for all $1 < j \leq r$ and 
\item $\proj{L_1} (f_n(\bar x)) = \begin{cases} 0 &\text{ if } \exists i: \proj{L_1}(x_i) \in H_1 \\
l &\text{ else},\end{cases}$
\end{itemize}

In the case that $a \notin H_1$, we simply set 
$$t_n = f_n (s_n (x_1,\ldots,x_{n^{k-2}}),\ldots,s_n (x_{n^{k-2}(n-1)+1}\ldots,x_{n^{k-1}})).$$ By construction $t_n$ has the desired property with respect to the wreath product representation $\bigotimes_{j =1}^m \algL_j \otimes \bigotimes_{i = 1}^r \algK_i \otimes \algM \otimes \algA/\rho$ and is of size $2^{O(n)} + n 2^{O(n)} = 2^{O(n)}$. In the case that $a \in H_1$, note that, since $\algK_1$ is generated by $a$ as a module, there must be an $e \in \Pol^{1}(\algK_1)$ with $e(a) \notin H_1$ and $e(0)=0$. We then use $e \circ s_n$ instead of $s_n$. This finishes the proof of Claim 1.\end{proof}

%Note that $t_n$ are non constant absorbing polynomials.
We are next going to use the circuits $t_n$ to reduce $p$-COLOR to $\CSAT(\algA)$ in subexponential time. For this let $G = (V,E)$ be an undirected graph. We are going to model the vertices $v\in V$ by terms $h_v := h(\bar x_v)$, where $h$ is the polynomial defined above, and $\bar x_v$ is a tuple of distinct variables (different for every $v$). By construction of $h$, the image of $h_v$ is in $[0]_\rho$, and the projection of $h_v$ to $M$ is surjective.

Without loss of generality let us assume that the graph $G$ has $|E| = n^{k-1}$ many edges, for some $n$ (otherwise we count some of the edges multiple times), and that $G$ is connected (and hence $|V|\leq |E|+1$). Then, let us consider the polynomial $t_{n} ( (h_v / h_w)_{(v,w) \in E})$ (where '$/$' denotes the right division of the loop multiplication). With respect to the wreath product representation $\algA = \bigotimes_{i = 1}^m \algL_i \otimes \algM \otimes \algA/\rho$, this polynomial evaluates to $(0,\ldots,0)$ if there is a pair $(h_v,h_w)$, such that $\proj{M}(h_v)$ and  $\proj{M}(h_w)$ are in the same coset of $H$ and $(l,0,\ldots,0)$, otherwise. So $(V,E)$ is $p$-colorable if the equation $t_{|E|} ( (h_v / h_w)_{(v,w) \in E}) = (l,0,\ldots,0)$ has a solution, and it is not $p$-colorable if the identity $t_{|E|} ( (h_v / h_w)_{(v,w) \in E}) = (0,\ldots,0)$ holds. 

The circuits constructed in this reduction have size $2^{O(n)}$, where $n = |E|^{(k-1)^{-1}}$. So $p$-coloring a graph $G$ can be reduced to $\CSAT(\algA)$ (and $\CEQV(\algA)$) in sub-exponential time $2^{c\cdot |E|^{(k-1)^{-1}}}$ for some $c>0$. If the exponential time hypothesis holds, then by Theorem \ref{theorem:ETC} (2) there is a lower bound of $2^{o(|E|+|V|)}=2^{o(|E|)}$ on the time complexity of $p$-coloring a connected graph. Thus, our instance of size $S = 2^{c\cdot |E|^{(k-1)^{-1}}}$ of $\CSAT(\algA)$ cannot be solved faster than in time $2^{o(|E|)}= 2^{o(\log(S)^{k-1})}$. The same result holds for $\CEQV(\algA)$; this finishes the proof for the case $\exp(\algM) = p > 2$.\\

In the case $\exp(\algM) = p = 2$, we similarly construct a sub-exponential reduction of 3-SAT to $\CSAT(\algA)$ (respectively $\CEQV(\algA)$). Such a reduction can be obtained, if, instead of the circuits $t_n$ in Claim 1, we use circuits $t_n'$ as in the following claim:\\
 
\textbf{Claim 2:} There is an elementary wreath product representation $\algA = \bigotimes_{i = 1}^m \algL_i \otimes \algM \otimes (\algA/ \rho)$, a hyperplane $H < (M,+,0,-)$ and an element $l \in L_1$, such that for every $n \in \N$, there are circuits $t_n'(\bar x)$ of arity $3n^{k-1}$, and size $2^{O(n)}$ such that for values of $\bar x$ in $([0]_\rho)$:
\begin{itemize}
\item $\proj{M} (t_n'(\bar x)) = 0$, $\proj{L_j} (t_n'(\bar x)) = 0$ for all $1 < j \leq m$ and 
\item $\proj{L_1} (t_n'(x_{1,1},x_{1,2}, x_{1,3},x_{2,1},\ldots, x_{n^{k-1},3})) = \begin{cases} 0 &\text{ if } \exists i: \proj{M}(x_{i,1}),\proj{M}(x_{i,2}), \proj{M}(x_{i,3}) \in H \\
l &\text{ else.}\end{cases}$\\
\end{itemize}

Except for the base case $k=2$ (which follows from Proposition \ref{proposition:main}), the inductive construction is the same as for $t_n$ in Claim 1. We leave the proof of Claim 2 to the reader (and refer also to the analogous construction in \cite[Section 3.2]{IKK-intermediate}). We can then construct a reduction from 3-SAT to $\CSAT(\algA)$ respectively $\CEQV(\algA)$ by identifying every variable $v$ of an instance with the polynomial $h_v = h(\bar x_v)$ and every negated variable $\neg v$ with $h_{\neg v} = b\cdot h(\bar x_v)$, where $b$ is an element of $[0]_{\rho} \setminus H$. Then, a 3-CNF $(L_{1,1} \lor L_{1,2} \lor L_{1,3}) \land (L_{2,1} \lor L_{2,2} \lor L_{2,3}) \land \cdots $  with literals $L_{i,j}$ is satisfiable if and only if $t_n'(h_{L_{1,1}}, h_{L_{1,2}}, h_{L_{1,3}}, h_{L_{2,1}},\ldots) = (l,0,\ldots,0)$ has a solution in $\algA$. It is unsatisfiable if and and only if the identity $t_n'(h_{L_{1,1}}, h_{L_{1,2}}, h_{L_{1,3}}, h_{L_{2,1}},\ldots) \approx (0,0,\ldots,0)$ holds in $\algA$. As in the previous case, this gives rise to a subexponential reduction of 3-SAT to $\CSAT(\algA)$, respectively $\CEQV(\algA)$. Together with Theorem \ref{theorem:ETC} (1) gives us the lower bounds of $2^{o(\log(S)^{k-1})}$ for $\CSAT(\algA)$ and $\CEQV(\algA)$ under the exponential time hypothesis.
\end{proof}

\section{Discussion} \label{sect:open}
We finish by discussing the two main directions for generalizations of Theorem \ref{theorem:main}.

As shown in \cite{IKKW-solvablegroups}, $\PolSAT(\mathbf G)$ for finite groups $\mathbf G = (G,\cdot,1,^{-1})$ of Fitting length $k$, exhibits the same quasipolynomial lower bounds as $\CSAT(\algA)$ for nilpotent algebras of Fitting length $k$ (assuming ETH).  This begs the question, if the lower bounds in Theorem \ref{theorem:main} even hold for $\PolSAT(\algA)$.  Or, more generally,  one may ask if $\PolSAT(\algA)$ has quasipolynomial lower bounds for all finite \emph{solvable} Mal'tsev algebras, which would subsume both the results of this paper and \cite{IKKW-solvablegroups}.  And, in fact,  in the conference paper \cite{IKK-Maltsev}, Idziak, Kawa\l{}ek and Krzaczkowski proved exactly this. Their proof uses tame congruence theory to construct \emph{polynomials} with properties similar to the \emph{circuits} $t_n$ in our proof of Theorem \ref{theorem:main}. 

This leads to the question, whether also in our proof of Theorem \ref{theorem:main}, it is enough to consider polynomials.  An analysis shows that the only point, in which we need to be careful in substitute circuits by polynomials is Proposition \ref{proposition:main}; in it we used circuits to efficiently express the iterated products with respect to the loop operation $+$.  Thus, if $+$ (or the Mal'tsev operation $m(x,y,z)$) is also part of the signature of the nilpotent algebra $\algA$, then the lower bounds of Theorem \ref{theorem:main} even hold for $\PolSAT(\algA)$.  Also, for general $\algA$ one can find a workaround by using $+$-terms of depth $\log(n)$ to define the ``sum'' of $n$ variables (using balanced trees), and then substitute every $+$ by the polynomial defining it (we would like to thank the anonymous reviewer of a previous version of this paper for this suggestion). However,  as \cite{IKK-Maltsev} deals with this problem anyways, we decided not to complicate our proof by adding this further addition.

Secondly,  recall that Theorem \ref{theorem:main} almost completly finishes the complexity classification of $\CSAT$ and $\CEQV$ for finite Mal'tsev algebras, assuming ETH and Conjecture \ref{conjecture:BST}. By Corollary \ref{corollary:main}, nilpotent Mal'tsev algebras of Fitting length $k \leq 2$ are the only algebras from congruence modular varieties, for which $\CSAT(\algA)$ and $\CEQV(\algA)$ could be tractable under the exponential time hypothesis.  While they are both known to be tractable in the supernilpotent case ($k=1$) \cite{kompatscher-EQSAT-supernilpotent, AM-commutator},  the classification for $k=2$ is still open. There are some tractability results for 2-nilpotent algebras \cite{IKK-examples, KKK-CEQV-2nilpotent}, i.e. algebras that are the wreath product of two affine algebras.  We further know, that not all algebras of Fitting length 2 (and not even all 2-nilpotent algebras) have tractable $\CSAT$ problems: Examples with proper quasipolynomial lower bounds (assuming ETH) were constructed in \cite[Example 7.2.]{IKK-modularcircuits}. In fact, those examples indicate that polynomial time algorithms for $\CSAT$ might only be feasible for algebras $\algA$, in which $\gamma_1 = \bigwedge_{i\geq 1} [\mathbf 1_A,\ldots, \mathbf 1_A]_i$ has blocks of prime power cardinality (and direct products of such algebras).  On the other hand, it is still consistent with our current knowledge, that all circuit equivalence problems $\CEQV(\algA)$, for nilpotent Mal'tsev algebras of Fitting length $k = 2$ are tractable. 

\bibliographystyle{alpha}
\bibliography{global_mk}
\end{document}